\documentclass[11pt,pdfa]{article}
\usepackage[in]{fullpage}

\usepackage[shortlabels]{enumitem}
\usepackage{bm}
\usepackage{amsfonts}
\usepackage{amsmath}
\usepackage{amssymb}
\usepackage{xcolor}
\usepackage{graphicx}
\usepackage{amsthm}
\usepackage{qtree}
\usepackage{tree-dvips}
\usepackage{float}
\usepackage{hyperref}
\usepackage[nameinlink]{cleveref}
\usepackage{mathtools}
\usepackage{braket}
\usepackage{mathrsfs}
\usepackage{tikz}
\usepackage{qcircuit}
\usepackage{dsfont}
\usepackage{longfbox}
\usepackage{mdframed}
\usepackage{comment}
\usepackage{thm-restate}
\hypersetup{colorlinks={true},linkcolor={blue},citecolor=magenta}

\theoremstyle{plain}
\newtheorem{theorem}{Theorem}[section]
\newtheorem{lemma}[theorem]{Lemma}
\newtheorem{claim}{Claim}

\newtheorem{conjecture}[theorem]{Conjecture}
\newtheorem{corollary}[theorem]{Corollary}
\newtheorem{proposition}{Proposition}
\newtheorem{definition}[theorem]{Definition}
\newtheorem{remark}[theorem]{Remark}

\newtheorem{algorithm}{Algorithm}

\renewcommand{\cal}[1]{\mathcal{#1}}
\newcommand{\R}{\mathbb{R}}

\newcommand{\N}{\mathbb{N}}

\newcommand{\E}{\mathop{\mathbb{E}}}

\newcommand{\rv}[1]{\mathbf{#1}}
\newcommand{\reg}[1]{\mathsf{#1}}

\newcommand{\ot}{\otimes}

\newcommand{\abs}[1]{\left | #1 \right|}

\newcommand{\eps}{\epsilon}
\newcommand{\ip}[1]{\left \langle #1 \right \rangle}
\newcommand{\supp}{\mathrm{supp}}
\newcommand{\ans}{\mathsf{ans}}

\newcommand{\negl}{\mathsf{negl}}

\newcommand{\cO}{\mathcal{O}}
\newcommand{\cD}{\mathcal{D}}

\newcommand{\class}[1]{\mathsf{#1}}
\newcommand{\QCMA}{\class{QCMA}}
\newcommand{\QMA}{\class{QMA}}
\newcommand{\Components}{\textsc{Components}}
\newcommand{\poly}{\mathrm{poly}}
\renewcommand{\cal}[1]{\mathcal{#1}}
\newcommand{\Lyes}{\cal{L}_{\mathrm{yes}}}
\newcommand{\Lno}{\cal{L}_{\mathrm{no}}}
\newcommand{\Hmin}{H_{\min}}

\date{}
\title{$\QMA$ vs. $\QCMA$ and Pseudorandomness} 
\author{Jiahui Liu \\ \small{MIT} \and Saachi Mutreja \\ \small{Columbia}\and 
Henry Yuen \\ 
\small{Columbia}}

\begin{document}
\maketitle

\begin{abstract}
    We study a longstanding question of Aaronson and Kuperberg on whether there exists a classical oracle separating $\mathsf{QMA}$ from $\mathsf{QCMA}$. Settling this question in either direction would yield insight into the power of quantum proofs over classical proofs. We show that such an oracle exists if a certain quantum pseudorandomness conjecture holds. Roughly speaking, the conjecture posits that quantum algorithms cannot, by making few queries, distinguish between the uniform distribution over permutations versus permutations drawn from so-called ``dense'' distributions. 
    Our result can be viewed as establishing a ``win-win'' scenario: \emph{either} there is a classical oracle separation of $\mathsf{QMA}$ from $\mathsf{QCMA}$, \emph{or} there is quantum advantage in distinguishing pseudorandom distributions on permutations.
\end{abstract}
\section{Introduction}
In 2007, Aaronson and Kuperberg showed that there is a quantum oracle $U$ relative to which $\QMA^U \neq \QCMA^U$~\cite{aaronson2007quantum}. Since then, finding a \emph{classical} oracle to separate the two complexity classes has been a fascinating question in quantum complexity theory. There has been progress on this question by considering either non-standard models of oracles or restricted $\QCMA$ verifiers: 
\begin{enumerate}
\item Fefferman and Kimmel~\cite{fefferman2018quantum},  Natarajan and Nirkhe~\cite{natarajan2023distribution}, and Li, Liu, Pelecanos and Yamawaka~\cite{li2024classical} separate $\QMA$ from $\QCMA$ relative to a \emph{distributional} oracle, in which the oracle is sampled from a distribution and the quantum proof can only depend on the distribution. 

\item Liu, Li, Pelecanos, and Yamakawa~\cite{li2024classical} also obtain the separation relative to a classical oracle, but under the restriction that the quantum verifier's query register is measured prior to querying the oracle. Ben-David and Kundu~\cite{ben2024oracle} recently showed that the separation under the restriction that the verifier makes a bounded number of rounds of adaptive queries to the classical oracle. 
\end{enumerate}

We continue ``chipping'' away at the $\QMA$ vs. $\QCMA$ question by proving the following:
\begin{theorem}[Main theorem, informal]
\label{thm:main-informal}
    Assuming a quantum pseudorandomness conjecture (described below), there exists a classical oracle $F$ such that $\QMA^F \neq \QCMA^F$.
\end{theorem}
\noindent In other words, we prove a \emph{conditional} separation relative to a standard classical oracle and without any restrictions on the verifier.

\subsection{A quantum pseudorandomness conjecture}

To explain the quantum pseudorandomness conjecture, we first explain the notion of a \emph{dense distribution}. We say that a random variable $\rv{F}$ over $\{0,1\}^N$ is \emph{$\delta$-dense} if for all subsets $S \subseteq [N]$ of coordinates, the marginal distribution of $F \sim \rv{F}$ restricted to the coordinates in $S$ has min-entropy at least $(1 - \delta)|S|$.\footnote{In the theoretical computer science literature this condition is often called ``blockwise dense,'' but we use ``dense'' for brevity.} In other words, every subset of coordinates has near maximal min-entropy. 

We now present a simplified version of our quantum pseudorandomness conjecture.

\begin{conjecture}[Simplified quantum pseudorandomness conjecture]
\label{conj:simplified}
Let $A$ be a quantum query algorithm that makes $T$ queries to an oracle $F \in \{0,1\}^N$. Let $\rv{F}$  denote a $\delta$-dense random variable over $\{0,1\}^N$. Then
\[
    \Big | \Pr_{F \sim \rv{U}}[A^F = 1] - \Pr_{F \sim \rv{F}}[A^F = 1] \Big| \leq \poly(T) \cdot \poly(\delta)
\]
where $\rv{U}$ denotes a uniformly random $N$-bit string, and the probabilities are over both the randomness of the algorithm $A$ as well as the randomness of sampling $F$.
\end{conjecture}
In other words, even though a dense random variable $\rv{F}$ may be far from uniform\footnote{An example of a dense distribution is the uniform distribution over even-parity strings: this is $\frac{1}{N}$-dense, but is $\Omega(1)$ far in statistical distance from the uniform distribution over all strings.}, no quantum algorithm can distinguish it from uniform by making few queries to an oracle sampled from it. In other words, the conjecture posits that dense distributions look \emph{pseudorandom} to all polynomial-query quantum algorithms. 

The conjecture is true for classical probabilistic query algorithms; this was essentially proved by Unruh~\cite{unruh2007random} (and later improved by Coretti, Dodis, Guo, and Steinberger~\cite{coretti2018random}) in the context of studying the random oracle model (from cryptography) when the adversary is provided some short advice about the random oracle. Guo, Li, Liu, and Zhang~\cite{guo2021unifying} first formulated the quantum pseudorandomness conjecture above and showed that it is closely linked to the \emph{Aaronson-Ambainis conjecture}, formulated by Aaronson and Ambainis in~\cite{aaronson2009need}:
\begin{conjecture}[Aaronson-Ambainis conjecture~\cite{aaronson2009need}]
\label{conj:aa-intro}
Let $p: \{0,1\}^N \to [0,1]$ be a degree-$d$ real polynomial. Then there exists a coordinate $j \in [N]$ such that $\mathrm{Inf}_j(p) \geq \poly(\mathrm{Var}(p)/d)$.
\end{conjecture}
\noindent Here, $\mathrm{Inf}_j(p)$ denotes the influence of the $j$'th variable in $p$ and $\mathrm{Var}(p)$ denotes $p$'s variance over the uniform distribution. 

The link between \Cref{conj:simplified} and the Aaronson-Ambainis conjecture is that they both imply that every quantum query algorithm can be classically simulated on \emph{most} inputs (with respect to the uniform distribution) with only a polynomial overhead.  
Just like the $\QMA$ vs. $\QCMA$ question, the Aaronson-Ambainis conjecture is also one of the central open problems in quantum complexity theory. Despite significant attention and a slew of partial results~\cite{o2015polynomial,arunachalam2019quantum,bansal2022influence,gutierrez2023influences,lovett2023fractional,bhattacharya2024aaronson}, the question remains open.

The difficulty of resolving the Aaronson-Ambainis conjecture (and its conclusion about the classical simultability of quantum query algorithms) suggests that proving \Cref{conj:simplified} will also be quite challenging. However there is some evidence that the pseudorandomness conjecture is true; in \Cref{sec:pseudorandomness} we prove a special case of it when both the density parameter $\delta$ and the number of queries $T$ is very small.

\begin{restatable}[Weaker version of \Cref{conj:simplified}]{lemma}{lemweaker}
\label{lem:weaker}
For all $\eps > 0$, for all $T$-query algorithms $A$, and for all $\delta$-dense random variables $\rv{F}$ over $\{0,1\}^N$, we have
\[
    \Big | \Pr_{F \sim \rv{U}}[A^F = 1] - \Pr_{F \sim \rv{F}}[A^F = 1] \Big| \leq \sqrt{\delta} 2^{O(T)}/\eps + O \Big( (\delta N)^{T/2} \eps \Big)~.
\]
\end{restatable}

For example, by setting $T = \frac{1}{4}\sqrt{\log N}$, $\delta = 2^{\sqrt{\log N}}/N$, and $\eps = N^{-1/3}$, \Cref{lem:weaker} implies that $T$-query algorithms cannot distinguish between $\delta$-dense distributions and the uniform distribution with more than $O(\sqrt{\delta}/\eps + N^{1/4 - 1/3}) = o(1)$ bias. 

This lemma is analogous to a weaker version of the Aaronson-Ambainis conjecture which establishes the existence of a variable with influence at least $\poly(\mathrm{Var}(p)/2^{O(d)})$, i.e., there is an exponential dependence on the degree of the polynomial. This weaker version was observed by Aaronson and Ambainis in~\cite{aaronson2009need} to follow from a powerful Fourier-analytic result of Dinur, Friedgut, Kindler, and O'Donnell~\cite{dinur2006fourier}.

\paragraph{The pseudorandomness conjecture for random permutations.} We now state the version of the quantum pseudorandomness conjecture that is needed for our $\QMA$ vs. $\QCMA$ separation result. First, recall the  \emph{permutation oracle} model, in which the quantum algorithm can query a permutation $P: [N] \to [N]$ and its inverse $P^{-1}$; that is, the algorithm can query the unitary operations $U_P: \ket{x,y} \mapsto \ket{x,y \oplus P(x)}$ and $U_{P^{-1}}: \ket{x,y} \mapsto \ket{x,y \oplus P^{-1}(x)}$. 

We furthermore consider the setting where the quantum algorithm can query a number of permutations $P_1,\ldots,P_R$, where each $P_r$ is a permutation on $[N_r]$ for some integer $N_r$. To query permutation $r \in [R]$ (resp., its inverse), the algorithm calls a unitary $U_{P_r}: \ket{x,y} \mapsto \ket{x,y \oplus P_r(x)}$ (resp., the unitary $U_{P_r^{-1}}: \ket{x,y} \mapsto \ket{x,y \oplus P_r^{-1}(x)}$).

\begin{restatable}[Quantum pseudorandomness conjecture for random permutations]{conjecture}{pseudorandomness}
\label{conj:pseudorandomness}
Let $A$ be a quantum query algorithm that makes $T$ queries to a collection of permutations $(P_1,\ldots,P_R)$ where each $P_r$ permutes $[N_r]$ for some integer $N_r$. Let $\rv{P}$ denote a random variable over $S_{N_1} \times \cdots \times S_{N_R}$ such that both $\rv{P}$ and the inverse permutation $\rv{P}^{-1}$ are $\delta$-dense. Then
\[
    \Big | \Pr_{P \sim \rv{Z}}[A^P = 1] - \Pr_{P \sim \rv{P}}[A^P = 1] \Big| \leq \poly(T, \log N_1,\ldots,\log N_R) \cdot \poly(\delta)
\]
where $\rv{Z}$ denotes a tuple of uniformly random permutations sampled from $S_{N_1} \times \cdots \times S_{N_R}$. Here the $\poly(\cdot)$ notation denotes universal polynomials that are independent of $N_1,\ldots,N_r,T,\delta$ (but can depend on $R$).
\end{restatable}

A sample of the random variable $\rv{P}$ over the product set $S_{N_1} \times \cdots \times S_{N_R}$ can be thought of as a string of length $N_1 + \cdots + N_R$ whose coordinates are indexed by pairs $\bigcup_{r=1}^R \bigcup_{j=1}^{N_r} \{ (r,j) \}$. We say that $\rv{P}$ (resp. $\rv{P}^{-1}$) is $\delta$-dense if for every subset $J \subseteq \bigcup_{r=1}^R \bigcup_{j=1}^{N_r} \{ (r,j) \}$, the min-entropy of $\rv{P}$ (resp. $\rv{P}^{-1}$) restricted to the elements in $J$ is at least
\[
    (1 - \delta) \sum_{r=1}^R \log \Big ( N_r(N_r-1) \cdots (N_r-|J_r|+1) \Big),
\]
where $J_r = J \cap \{ (r,j) \}_j$. In other words, the min-entropy is $1 - \delta$ fraction of its maximal value. 

 The condition that the inverse permutation $\rv{P}^{-1}$ is dense is necessary for the conjecture to hold. Consider the following example: $\rv{P}$ is a uniformly random permutation over $[N]$ conditioned on $\rv{P}^{-1}(1)$ being an even integer. The forward permutation $\rv{P}$ is $\delta$-dense for $\delta = O(1/N)$, but clearly $\rv{P}$ is easily distinguishable from a uniformly random permutation by querying $\rv{P}^{-1}(1)$.\footnote{We are indebted to Mark Zhandry  for pointing out this example.} In this example the inverse permutation $\rv{P}^{-1}$ is not $\delta$-dense for small $\delta$, because the entropy of $\rv{P}^{-1}(1)$ is $\frac{1}{2} \log N$.  

In this conjecture, the number $R$ of permutations is treated as a fixed constant, while the $N_r$'s are numbers growing to infinity and $T,\delta$ are allowed to depend on the $N_r$'s. \\
\medskip

The main difference from \Cref{conj:simplified} is that the algorithms in \Cref{conj:pseudorandomness} query random permutations, rather than random boolean functions. Furthermore, the model of oracle algorithm is different, in that $A$ can compute the permutation both the ``forwards'' as well as ``inverse'' directions. 
Thus it is not clear if it is comparable with \Cref{conj:simplified} (i.e., whether one implies the other). 

The main pieces of evidence for \Cref{conj:pseudorandomness} come from (a) its resemblance to the random oracle model formulation of \Cref{conj:simplified}, and (b) a proof for the special case of classical probabilistic query algorithms; this was shown by Coretti, Dodis, and Guo~\cite[Claim 18]{coretti2018non} in the context of studying the power of auxiliary information in the random permutation model in cryptography.\footnote{We note that the statement of Claim 18 in~\cite{coretti2018non} does not explicitly state the condition that the inverse permutation is dense. However, as illustrated by the example above, this is necessary, and their proof implicitly uses this condition.} It is an interesting question to find additional evidence for \Cref{conj:pseudorandomness} (such as proving a permutation analogue of \Cref{lem:weaker}) -- or refuting it. 

Looking ahead, we can interpret our main result (\Cref{thm:main-informal}) as spelling out a ``win-win'' scenario: either quantum proofs are more powerful than classical proofs (i.e., there is a classical oracle separating $\QMA$ from $\QCMA$), or there is quantum advantage in distinguishing between dense permutation distributions from the uniform distribution (or both could be true!). The possibility of quantum advantage for this task would be rather interesting, as it could have implications for the post-quantum (in)security of various ciphers and hashes, which are often modeled with random permutations (see~\cite{coretti2018non} for more discussion about this).

\subsection{The conditional oracle separation}
\label{sec:intro_oracle_description}

We now describe the oracle problem we use for the conditional $\QMA$ vs. $\QCMA$ separation. Informally, the oracle is promised to encode an exponentially-large constant-degree graph that has either (at least) two connected components, or is a connected expander graph; the task is to distinguish between the two cases. This is a slight simplification of the problem considered by Natarajan and Nirkhe~\cite{natarajan2023distribution} in their $\QMA$ vs. $\QCMA$ separation in the randomized oracle model, except they consider distinguishing between having (exponentially) many connected components and having a single connected component.

\paragraph{The graph oracle model.} First, we specify the oracle  model. Fix integers $n,R \in \N$. Let $N = 2^n$. The oracles we consider are functions $F: [R] \times [N] \to [N]$ such that for all $r \in [R]$, $F(r,\cdot)$ describes a perfect matching on $[N]$: $F(r,x) = y$ if and only if $F(r,y) = x$. We call such a function $F$ a \emph{graph oracle} because it naturally corresponds to a (multi)graph $G_F$ with vertex set $[N]$ and there is an edge $(x,y)$ if there exists an $r$ such that $F(r,x) = y$. In fact, the graph oracle $F$ encodes more information in that each edge $(x,y)$ has a \emph{color} corresponding to the variable $r$ such that $F(r,x) = y$. 

A quantum query algorithm can query a graph oracle $F$ via the unitary $U_F : \ket{r,x,u} \mapsto \ket{r,x,u \oplus F_r(x)}$ where  $\oplus$ denotes addition over the binary field (we identify $[N]$ with binary strings $\{0,1\}^n$). 

In this paper we think of $R$ as a fixed constant (indeed, $R=3$ suffices) and $N$ as the growing parameter. Let $\Delta > 0$ be such that a random degree-$R$ graph formed by taking the union of $R$ random perfect matchings of $[N]$ yields a $\Delta$-expander (i.e., the smallest nonzero eigenvalue of the normalized graph Laplacian is at least $\Delta$) with high probability (see \Cref{lem:random-graphs}). 

\paragraph{The oracle problem.} Now we formally define the oracle problem, which we call $\Components$. Consider two sets $(\Lyes,\Lno)$ of graph oracles where
\begin{itemize}
    \item $F \in \Lyes$ if and only if there exists a subset $S \subseteq [N]$ with $|S| = N/2$ such that there are no edges between $S$ and $\overline{S}$ in $G_F$. 

    \item $F \in \Lno$ if and only if the graph $G_F$ is a $\Delta$-expander. %
\end{itemize}
The $\Components$ problem is to distinguish whether a given graph oracle $F$ is in $\Lyes$ or $\Lno$, promised that one is that case. 

Natarajan and Nirkhe~\cite{natarajan2023distribution} show that the $\Components$ can be solved in the $\QMA$ model: given oracle access to an instance $F \in \Lyes$, the prover can send the verifier the following $\log N$-qubit state:
\[
    \ket{\psi} = \sqrt{\frac{1}{N}}\sum_{x \in S} \ket{x} - \sqrt{\frac{1}{N}} \sum_{x \notin S} \ket{x}~.
\]
The quantum verifier performs one of two tests: 
(a) check the state is orthogonal to the uniform superposition over \emph{all} vertices, or (b) check the state is invariant under permuting the vertices according to $F$. Since $F \in \Lyes$, the verifier will accept with probability $1$. On the other hand, if $F \in \Lno$, then it can be shown that \emph{no} state can cause the verifier to accept with high probability. This verification procedure runs in time $\poly(\log N)$.

Can the $\Components$ problem be efficiently solved by a quantum verifier given a classical witness? 
Natarajan and Nirkhe showed the following ``distribution testing'' lower bound: all $\QCMA$ proof systems must fail at solving (their version of) the $\Components$ problem when the verifier only makes $\poly(\log N)$ queries to the oracle and the witness length is $\poly(\log N)$ -- under the assumption that the classical witness only depends on the \emph{partition} of vertices of the underlying graph $G_F$ into connected components, and not on how the vertices are connected (i.e. the edges).

For our conditional separation, we trade the assumption about the witness structure for a plausible assumption about the pseudorandomness of dense distributions.

\subsection{Our approach}  
The proof of \Cref{thm:main-informal} proceeds in two steps. First, assume that there exists an efficient $\QCMA$ verifier for the $\Components$ problem. We show that \Cref{conj:pseudorandomness} implies that, without loss of generality, the witness specifies a $\poly(\log N)$-sized list of vertices and edges in the graph. That is, the witness has \emph{some} amount of structure, albeit less structure than what is assumed in~\cite{natarajan2023distribution}. %

In the second step, we show that such a witness cannot help distinguish between Yes and No instances of $\Components$. Intuitively, the verifier should not be able to tell by making $\poly(\log N)$ queries to the oracle whether the specified vertices and edges are part of the same component or belong to different components. 

\paragraph{Applying the pseudorandomness conjecture.}
The connection between our oracle structure and pseudorandomness in the first step of the proof arises as follows. If the classical witness is short (e.g., $\poly(\log N)$ length), then by a standard counting argument, there must exist a set of Yes instances corresponding to a \emph{maximizing witness} $w^*$, one which is associated with the maximum number of Yes instances such that the $\QCMA$ verifier accepts. We call the set of these Yes instances $\cal{L}_{w^*}$. Now, hardwiring  $w^*$ into the supposed $\QCMA$ verifier solving the $\Components$ problem, we can obtain a $\poly(\log N)$-query algorithm that does \emph{not} use a witness but can distinguish between instances in $\cal{L}_{w^*}$ and instances in $\Lno$. Since the maximizing witness $w^*$ is short, the set $\cal{L}_{w^*}$ occupies at least $2^{-\poly(\log N)}$ fraction of $\Lyes$, and therefore the uniform distribution over $\cal{L}_{w^*}$ has min-entropy at least $\log |\Lyes| - \poly(\log N)$. This argument is standard in all previous works studying $\QMA$ vs $\QCMA$ separations~\cite{aaronson2007quantum,fefferman2018quantum,natarajan2023distribution,li2024classical}.

It is not clear at a first sight why we can apply a conjecture on pseudorandom permutation oracles to this high min-entropy distribution over graph oracles. To do this, we introduce a ``raw permutation'' model of viewing each graph oracle $G \in \Lyes\cup \Lno$.  In more detail, we show that any algorithm querying a $\Components$ oracle can be viewed as an algorithm that queries a small number of independently-chosen permutations $X,Y_1,\ldots,Y_R,Z_1,\ldots,Z_R$. The permutation $X:[N] \to [N]$ partitions $[N] = S \cup \overline{S}$ by letting $S$ be the image of $\{1,2,\ldots,N/2\}$ under $X$, and $\overline{S}$ the complement. Then, for each $r \in [R]$, the permutations $Y_r,Z_r:[N/2] \to [N/2]$ specifies a matching on $S$ and $\overline{S}$, respectively. 

Furthermore, the uniform distribution over raw permutations induces a uniform distribution over corresponding graph oracles in $\Lyes$. Conversely, the uniform distribution over $\cal{L}_{w^*}$ induces a high-min entropy distribution over the set $\cal{P}_{w^*}$ of all ``raw permutations'' $X Y_1 \cdots Y_R Z_1 \cdots Z_R$ that give rise to a graph oracle $F \in \cal{L}_{w^*}$. By a well-known decomposition lemma for high min-entropy distributions from complexity theory~\cite{goos2016rectangles,kothari2021approximating} and cryptography~\cite{coretti2018random,coretti2018non,guo2021unifying}, the uniform distribution over $\cal{P}_{w^*}$ can be decomposed into a convex combination of distributions over random permutations such that both the forward and inverse distributions are $(k, \delta)$-dense, for $k = \poly(\log N)$ and $\delta = \poly(\log N)/N$. Here, a $(k,\delta)$-dense distribution is one that is fixed to a deterministic value in some $k$ coordinates, but is $\delta$-dense everywhere else.

Then, using Conjecture \ref{conj:pseudorandomness}, each dense source in the convex combination is indistinguishable (by an efficient query algorithm) from a \emph{bit-fixing source}. In the context of raw permutation model, a bit-fixed raw permutation would imply that the $X$, $\{Y_r\}_{r \in [R]}$ and $\{Z_r\}_{r \in [R]}$ permutations are fixed in a small number of coordinates in $[N]$ or $[N/2]$, and are uniform otherwise. Translated to graph oracles, these fixings correspond to at most $\poly(\log N)$ fixed vertices and edges; let $\rho$ denote this \emph{graph fixing data}. Let $\Lyes^\rho$ and $\Lno^\rho$ denote the graph oracles in $\Lyes$ and $\Lno$ respectively that are consistent with $\rho$. The verifier with the hardwired witness $w^*$ thus can distinguish between $\Lyes^\rho$ and $\Lno^\rho$ with few queries.

\paragraph{Uselessness of the fixed coordinates.}
The second step is to show that revealing $\poly(\log N)$-size graph fixing data $\rho$ cannot help a quantum algorithm distinguish between Yes and No instances of the $\Components$ problem.
We prove via a reduction: we show that if there was an efficient query algorithm $A$ that distinguishes between $\Lyes^\rho$ and $\Lno^\rho$, then it can be transformed into an efficient query quantum algorithm $M$ that solves the original $\Components$ problem (without any witness).

Our reduction takes as input a uniformly random graph oracle $F \in \Lyes \cup \Lno$, and \emph{plants} the graph fixing data $\rho$ in the graph $F$ by only making $\poly(\log N)$ queries to $F$. Let $F'$ denote the resulting graph oracle. Then, it simulates the algorithm $A$ with the oracle $F'$. When $F\in \Lno$, then $F' \in \Lno^\rho$ and therefore $A$ would reject with high probability. 

The main challenge is to prove that when $F$ is sampled uniformly at random from $\Lyes$, the distribution over $F'$ is uniform over $\Lyes^{\rho}$. Once this is shown, then we deduce that there is a $\poly(\log N)$-query algorithm for the $\Components$ problem, which contradicts
(a slightly modified version of) the query lower bound of Ambainis, Childs, and Liu~\cite{ambainis2011quantum} on distinguishing whether a given graph oracle represents an expander or a graph with two connected components. We note that this part of the proof is directly inspired by the approach of Natarajan and Nirkhe~\cite{natarajan2023distribution}, who also reduce to the same query lower bound of~\cite{ambainis2011quantum}. 

Putting everything together, there cannot be an efficient $\QCMA$ verifier for the $\Components$ problem, assuming the pseudorandomness conjecture.

\subsection{An unconditional $\QMA$ vs. $\QCMA$ separation for interactive games and distributional oracles}

 In addition to the previous conditional oracle separation result, we 
 give a simple proof of an \emph{interactive version} of $\QMA/\QCMA$ separation, which directly leads to a simplified proof of the distributional oracle separation in ~\cite{natarajan2023distribution,li2024classical}. 
Interactive games are a useful notion in cryptography and complexity theory. The techniques employed here to analyze the lower bound for interactive games have been previously utilized widely in addressing the non-uniformity of algorithms -- algorithms that have an unbounded preprocessing stage to produce an advice (in our case, a witness) that helps a bounded second stage algorithm to solve a  challenge. This formulation of interactive games offers a clear and alternative perspective on the distributional oracle separation explored in earlier works, highlighting their underlying conceptual similarity.

The oracle has a similar structure to the one used in the $\Components$ problem, but is slightly different. We then observe that this $\QMA/\QCMA$ separation for interactive games almost directly gives a $\QMA/\QCMA$ separation with a distributional oracle, reproving the results of~\cite{natarajan2023distribution,li2024classical}. %

We describe the high level idea of our interactive game separation. 
The oracle $F$ we give out still represents a disconnected graph with two highly connected components in the Yes-instance and a connected graph in the No-instance. But the difference is that each sub-oracle $F(r,\cdot)$ specifies a single-cycle random permutation instead of any permutation and there are exponentially many sub-oracles $F(r, \cdot)$. 

We first define an interactive variant of the $\QMA/\QCMA$ protocol.
In the interactive game, prover $P$ is first given access to the entire oracle $F(\cdot, \cdot)$, where each sub-oracle $F(r,\cdot)$ is either a Yes oracle: two disjoint cycles, each on $N/2$ vertices (the partitions are fixed for all Yes oracles); or one single cycle on all $N$ vertices. $P$ then produces a witness. When it comes to a challenge phase, the challenger will sample a random $r \gets [R]$. Now the verifier $V$ with the witness should output whether $F(r,\cdot)$ is a Yes-oracle or No-oracle.

The $\QMA$ upper bound is essentially the same as for the $\Components$ problem: a uniform superposition over all vertices except vertices in one component has a $1$ phase and vertices in the other component has a $-1$ phase. 

We then give an $\QCMA$ lower bound for this interactive variant through a technique that removes the use of witness in our analysis with a sequentially repeated game (multi-instance security). Suppose there is a polynomial-query $\QCMA$ verifer that can distinguish two oracles with a some arbitrary polynomial-size classical witness, then we can build a verifier that guesses a witness and succeeds in winning a sequentially repeated version of the interactive oracle distinguishing game, by breaking the sequential repetition soundness. Let us denote the interactive game as $G$ and the $k$-sequentially repeated game as $G^{\otimes k}$.

Since the individual games in $G^{\otimes k}$ are being played sequentially and each challenge is sampled independently with a fresh choice of $r_i$ for $F(r_i,\cdot)$ at each round, we can argue that the winning probability decays exponentially by a lemma in \cite{chung2020tight}.

What's left is to analyze the probability of success in a single instance $G$. This reduces to the problem of deciding whether a given permutation oracle is a single cycle or a product of two equal length cycles. This should be hard to do using a polynomial number of queries, and thus the success probability should be at most    $\frac{1}{2} + \eta$
for some negligible quantity $\eta$, through Ambainis's adversary method~\cite{ambainis2000quantum}.
Therefore, the success probability in $G^{\otimes k}$ is at most $(\frac{1}{2} + \eta)^k$.

Finally, we can leverage the loss in guessing the advice.
When we set the repetition number $k$ to be a proper superpolynomial in terms of $\log N$, then we can observe that the winning probability in a single-instance interactive game \emph{with advice} is also upper bounded by $1/2 + \negl(\log N)$.

In the end, we show a simple reduction from a distributional oracle QMA/QCMA game to our interactive QMA/QCMA game.

\subsection{Relation with prior and concurrent work}

We briefly compare our results and approach with prior work. The $\Components$ problem and variants of it were previously studied by Lutomirski~\cite{lutomirski2011component}, Fefferman and Kimmel~\cite{fefferman2018quantum}, Natarajan and Nirkhe~\cite{natarajan2023distribution}, and Bassirian, Fefferman, and Marwaha~\cite{bassirian2023power}.\footnote{In fact, it was pointed out to us by Kunal Marwaha that Lutomirski had already conjectured that the $\Components$ problem was a candidate to separate $\mathsf{QMA}$ from $\mathsf{QCMA}$ in~\cite{lutomirski2011component}.}

Our $\QCMA$ lower bound proof is a marriage of the proof of~\cite{natarajan2023distribution} and the proof approach of Li, Liu, Pelecanos, and Yamakawa~\cite{li2024classical}. The latter prove a $\QCMA$ lower bound assuming that the verifier only makes classical queries to the oracle. In this case, the verifier is essentially a classical randomized verifier (because we only consider query complexity, not time complexity). The problem they use is based on the Yamakawa-Zhandry problem, which was originally used to show verifiable quantum advantage in the random oracle model~\cite{yamakawa2024verifiable}. They argue that a $\QCMA$ verifier for their problem can be converted into an efficient classical query algorithm that distinguishes between a ``bit-fixed'' distributions over Yes instances from all No instances. This conversion uses the pseudorandomness conjecture for random functions (\Cref{conj:simplified}) for classical query algorithms, which was proved by~\cite{coretti2018random}. However, \cite{li2024classical} shows that such a classical query algorithm does not exist. It is not clear when allowing the $\QCMA$ verifier to make quantum queries in their setting, whether a  query lower bound can be obtained for their problem, even assuming a quantum bit-fixing conjecture like we do in this work.

Ben-David and Kundu~\cite{ben2024oracle} analyze $\QCMA$ verifiers that are restricted to making their queries with little adaptivity; this means that the queries are performed in parallel in a few layers (but there can be polynomially-many queries in each layer). They show that a variant of the Yamakawa-Zhandry-inspired problem from~\cite{li2024classical} does not admit such bounded-adaptivity $\QCMA$ verifiers. 

Finally, we mention a paper~\cite{zhandry2024towards} that recently appeared as we were preparing to post this manuscript on arXiv. Like us, Zhandry also proves a $\QMA$ vs. $\QCMA$ separation conditional on a conjecture. The conjecture, roughly speaking, stipulates that distributions that are \emph{approximately} $k$-wise independent (in a certain multiplicative error sense) can be ``rounded'' to a nearby distribution that is \emph{exactly} $k$-wise independent. Zhandry considers the problem of distinguishing whether a sparse subset $S \subseteq [N]$ has superpolynomial size or is empty, given membership queries to $S$ \emph{as well as} membership queries to the set $U$ of ``heavy'' elements sampled from the Fourier transform of a Haar-random state supported on $S$. \cite{zhandry2024towards} shows that there is a $\QMA$ upper bound for this problem, but assuming his statistical conjecture there is no $\QCMA$ verifier for the problem.

Our conjecture and Zhandry's conjecture share some similarities; his conjecture can also be interpreted as a pseudorandomness statement, as the result of~\cite{zhandry2015secure} implies that (exactly) $k$-wise independent distributions cannot be distinguished from the uniform distribution by few-query algorithms. On the other hand, it is not clear whether our conjectures are comparable (whether one implies the other).

\subsection{Next steps}

The main contribution of our paper is to connect the $\QMA$ vs. $\QCMA$ oracle separation question with a question about pseudorandomness against (query-)efficient quantum algorithms. This latter question is in turn deeply related to the Aaronson-Ambainis conjecture, which is one of the most important questions in quantum complexity theory. 
We believe this connection suggests several immediate avenues for exploration.

\begin{enumerate}
    \item What are the relationships between \Cref{conj:simplified} (the pseudorandomness conjecture in the random oracle model), \Cref{conj:aa-intro} (the Aaronson-Ambainis conjecture), and \Cref{conj:pseudorandomness} (the pseudorandomness conjecture in the random permutation model)? 
    For example, does \Cref{conj:simplified} imply \Cref{conj:pseudorandomness} or vice-versa?

    \item Can \Cref{conj:pseudorandomness} be deduced as a consequence of a simpler version of \Cref{conj:pseudorandomness} that only refers to a \emph{single} permutation oracle (i.e. $R = 1$)? 

    \item Can weaker versions of the quantum pseudorandomness conjecture (\Cref{conj:pseudorandomness}) be proved? For example, in analogy to \Cref{lem:weaker}, can the conjecture be proved where the dependence on $T$ is exponential (rather than polynomial)? Another question is whether the conjecture can be proved if the number of \emph{adaptive rounds} of queries is bounded? This would immediately give an alternate proof of the bounded-adaptivity lower bound of Ben-David and Kundu~\cite{ben2024oracle}. 

    \item The oracle separations of Li, Liu, Pelecanos, and Yamakawa~\cite{li2024classical}  and Ben-David and Kundu~\cite{ben2024oracle} are based on the Yamakawa-Zhandry problem~\cite{yamakawa2024verifiable}. An attractive feature of these problems is that the Yamakawa-Zhandry problem is based on random boolean functions (rather than on random permutations as the $\Components$ problem is). Can a full $\QMA$ vs. $\QCMA$ separation be obtained from a problem inspired by Yamakawa-Zhandry, assuming \Cref{conj:simplified} (the pseudorandomness conjecture in the random oracle model)? 

    \item Can it be shown that, for the specific oracle problems considered in this paper or in other works (e.g., the $\Components$ problem), a classical witness in a $\QCMA$ protocol effectively fixes a subset of coordinates of the oracle? In our paper we relied on the full power of the pseudorandomness conjecture to argue this, but it is plausible that this can be directly argued for the $\Components$ problem, thus avoiding having to resolve the Aaronson-Ambainis conjecture or related questions (although that is also a worthwhile outcome!). 
\end{enumerate}

\subsection*{Acknowledgments} 
 We thank Chinmay Nirkhe, Shalev Ben-David, Qipeng Liu, Kunal Marwaha, and Mark Zhandry for helpful discussions. 
 HY and SM are supported by AFOSR awards
 FA9550-21-1-0040 and FA9550-23-1-0363, NSF CAREER award CCF-2144219, and the Sloan Foundation.
 JL was supported by DARPA under Agreement No. HR00112020023, NSF CNS-2154149,
 and a Simons Investigator Award when starting this project. JL is currently at Fujitsu Research.

\section{Preliminaries}
\label{sec:prelim}

\subsection{Random variables}
We will use boldfaced letters $\rv{X},\rv{Y},\ldots$ to denote random variables. To denote a specific instantiation of a random variable $\rv{X}$, we will use non-boldfaced variables $X$. If $\rv{X}$ takes values in $\Sigma^n$ for some finite alphabet $\Sigma$, then for a subset $I \subseteq [n]$ we write $\rv{X}_I$ to denote the random variable corresponding to the restriction of $\rv{X}$ to the coordinates in $I$. The support of $\rv{X}$, denoted by $\supp(\rv{X})$, is the set of all values it can take with nonzero probability. The \emph{min-entropy} of the random variable $\rv{X}$ over a finite set, denoted by $\Hmin(\rv{X})$, is defined to be $\min_{X \in \supp(\rv{X})} \log \frac{1}{\Pr(\rv{X} = X)}$. %

\subsection{Graphs and expansion}

Let $G = (V,E)$ be an undirected graph where there can be multiple edges between a pair of vertices. Let $A$ denote the adjacency matrix of $G$, i.e., $A_{uv}$ is the number of edges between vertices $u,v \in V$. Suppose that $G$ is $R$-regular; i.e., each vertex has $R$ neighbors. Then the normalized Laplacian of $G$ is given by the matrix $\cal{L} = I - \frac{1}{R} A$. The normalized Laplacian is a positive semidefinite matrix and its maximum eigenvalue is at most $1$. The spectral gap $\Delta$ of $\cal{L}$ is its second smallest eigenvalue (which may be $0$ if the kernel of $\cal{L}$ has dimension greater than $1$). We say that a graph $G$ is a $\Delta$-expander if its spectral gap is at least $\Delta$. 

\begin{proposition}
Let $\cal{L}$ denote the normalized Laplacian of a graph $G$ with spectral gap $\Delta$. Then for all vectors $\ket{\psi} = \sum_{ v\in V} \alpha_v \ket{v}$, we have
\[
    \| \ket{\psi} - \ket{\overline{\psi}} \|^2  \leq \frac{1}{\Delta} \bra{\psi} \cal{L} \ket{\psi}
\]
where $\ket{\overline{\psi}}$ denotes the vector where every entry is $\frac{1}{|V|} \sum_v \alpha_v$.
\end{proposition}

The following is a well-known fact that relates the mixing time of a random walk on graph with its spectral gap. 
\begin{lemma}
    \label{lem:expander-mixing}
    Let $G$ be a $R$-regular $\Delta$-expander graph with $N$ vertices. Consider the following $t$-step random walk $\rv{W}$ starting at some vertex $v$ of the graph: at each time step stay in place with probability $1/2$, and move to a uniformly random neighbor with probability $1/2$. The probability distribution over the ending vertex of the walk $\rv{W}$ is $N \Delta^t$-close in total variation distance from the uniform distribution over the graph.
\end{lemma}
\begin{proof}
    A proof of this can be founded in, e.g.,~\cite{goldreich2011randomized}.
\end{proof}

\begin{corollary}
\label{cor:tlazyrw}
    Let $G$ be a $R$-regular $\Delta$-expander graph with $N$ vertices. Let $K$ be an integer, and consider a $Kt$-step random walk $\rv{W}$ as described in \Cref{lem:expander-mixing}. Let $\rv{E} = (\rv{e}_1,\ldots,\rv{e}_K)$ denote a sequence of vertices where $\rv{e}_k$ is the $kt$'th vertex in the walk $\rv{W}$. Then the distribution of $\rv{E}$ is $KN\Delta^t$-close to the uniform distribution over $V^K$, where $V$ is the vertex set of $G$.
\end{corollary}
\begin{proof}
    This follows applying a hybrid argument to \Cref{lem:expander-mixing}.
\end{proof}

\begin{lemma}[Union of $3$ random matchings is a good expander]
    \label{lem:random-graphs}
    There exists a universal constant $0 < \Delta < 1$ such that the following holds. 
    For all sufficiently large even integers $N$,
    \[
        \Pr[\text{$G$ is an $\Delta$-expander}] \geq 1 -\frac{1}{N}
    \]
    where $G$ is a graph formed by taking the union of $3$ independently chosen random matchings on $[N]$.
\end{lemma}
\begin{proof}
    See, e.g., Goldreich's lecture notes~\cite{goldreich2011randomized} for a proof.
\end{proof}

\subsection{Quantum query lower bounds}

We collect some quantum query lower bounds that will be used to prove both our $\QCMA$ lower bound (\Cref{sec:qcma-lower-bound}) and the interactive game lower bound (\Cref{sec:interactive-game}).

The first lower bound is the so-called ``hybrid method'' due to Bennett, Brassard, Bernstein, and Vazirani~\cite{BBBV97}; this roughly states that if a quantum query algorithm makes a few queries to an oracle, then the algorithm cannot ``notice'' small changes to most places in the oracle. We can model the state of a quantum query algorithm right before a query as some superposition
\[
    \ket{\phi} = \sum_{b,x,y,w} \alpha_{b,x,y,w} \ket{b,x,y,w}
\]
where $\ket{b}$ denotes a single-qubit control register, $\ket{x}$ denotes the query input register, $\ket{y}$ denotes the query target register, and $\ket{w}$ denotes the workspace. Define the following quantity, known as the \emph{query weight} of $A$ on input $x$:
\[
    W_x(\ket{\phi}) = \sum_{b,y,w} |\alpha_{b,x,y,w}|^2~.
\]
After the query to the (controlled) oracle $cO$ the state becomes
\[
    cO \, \ket{\phi} = \sum_{b,x,y,w} \alpha_{b,x,y,w} \ket{b,x,y \oplus bO(x),w}
\]
where by $bO(x)$ we mean the zero string if $b = 0$ and otherwise $O(x)$.

\begin{theorem}[\cite{BBBV97}] 
\label{thm:bbbv97_oraclechange}
Let $A$ denote a $T$-query quantum algorithm. Let $O: [N] \to [M]$ be an oracle and let $\ket{\phi_i}$ denote the intermediate state of $A$ right before the $i$'th query to $O$.  Let $F \subseteq [N]$ be a set of query inputs, and let $\widetilde{O}$ denote an oracle that is identical to $O$ on all inputs $[N]\setminus F$. Then
\[
    \Big \| \ket{\phi_{T+1}} - \ket{\widetilde{\phi}_{T+1}} \Big \| \leq \sqrt{T \sum_{x \in F} \sum_{i=1}^T W_x(\ket{\phi_i})}
\]
where $\ket{\phi_T},\ket{\widetilde{\phi}_T}$ are the final states of the algorithm $A$ when querying oracles $O$ and $\widetilde{O}$, respectively.  
\end{theorem}

\paragraph{Lower bound for testing expansion.} In our $\QCMA$ lower bound proof we convert a hypothetical $\QCMA$ verifier for the $\Components$ problem into a query-efficient algorithm for the $\Components$ problem (\emph{without} using a witness). However, Ambainis, Childs, and Liu~\cite{ambainis2011quantum} proved a $\Omega(N^{1/4})$ lower bound on quantum algorithms that can test whether an $N$-vertex bounded-degree graph is an expander or far from being an expander. Their lower bound is \emph{almost} what we need: they show that no algorithm making $o(N^{1/4})$ queries to the graph can distinguish between whether the graph is drawn from a distribution over two disconnected random graphs, or a single connected random graph. 

In more detail, their distribution over random graphs are chosen as follows. Let $M$ be divisible by $4$, and let $N$ be such that $M = N + \Theta(N^{0.9})$. A Yes instance as sampled in two steps.
\begin{enumerate}
    \item Partition $[M]$ into 2 equal sets $V_1,V_2$. For each set $V_j$, choose $R$ random perfect matchings on $V_j$ and take their union. Let $G'$ be the resulting graph.
    \item Pick $N$ vertices $v_1,\ldots,v_N$ at random as follows: for each $v_i$, pick $j \in \{1,2\}$ uniformly at random, and then sample $v_i$ uniformly at random from $V_j$, conditioned on not being sampled before. Let $G$ be the induced subgraph of $G'$ on the vertices $v_1,\ldots,v_N$. 
\end{enumerate}
A No instance is sampled similarly, except there is only one component. 

Without the subsampling step (the second step above), the distributions would be identical to the uniform distributions over $\Lyes$ and $\Lno$ respectively of the $\Components$ problem. However the subsampling step is helpful for the lower bound proof~\cite{ambainis2011quantum}, which employs the polynomial method of~\cite{beals2001quantum} in a sophisticated way. 

However it is more useful for us to have a lower bound for distinguishing between (the uniform distributions over) $\Lyes$ and $\Lno$. We deduce such a lower bound as a corollary of~\cite{ambainis2011quantum}, using the BBBV theorem (\Cref{thm:bbbv97_oraclechange}).

\begin{lemma}
\label{lem:modified_acl}
Let $A$ be a quantum algorithm that makes $o(N^{1/20})$ queries. Then
\[
    \Big | \Pr_{F \sim \Lyes} [A^F = 1] - \Pr_{F \sim \Lno} [A^F = 1] \Big| \leq \frac{1}{1000}.
\]
\end{lemma}

\begin{proof}

The main idea is that the subsampling step omits so few vertices (at most a $\Theta(N^{-0.1})$ fraction) that query-efficient algorithms should not be able to detect if the subsampling step was skipped. 

Suppose that any quantum algorithm $A$ is given oracle $F'$ corresponding to graph $G'$ in one world and given oracle $F$ for graph $G$ in another world (either in the Yes-instance case or in the No-instance case), we argue that $A$'s states in these two worlds have negligible difference in trace distance at the end of its algorithm. Therefore we can apply the \cite{ambainis2011quantum} lower bound to any algorithm with access to $F'$.

Ambainis, Childs, and Liu~\cite{ambainis2011quantum} use the same graph oracle model as we do: the oracle $F(r,u)$ returns the vertex that is matched with $u$ in the $r$'th matching. %
When we change an oracle $F'$ to $F$, we apply the following changes: (1) the vertices $v \in V$ are chosen randomly from $V' = [M]$ in the subsampling step. Let us denote the set of chosen vertices as $V$, where $\vert V \vert = N$. We first implement an intermediate oracle $F_0$.
For all $v \in V', v \notin V$ where $V'$ is the vertex set in $G'$, $F_0(r,v) = \bot$ for all $r$; for all $u = F'(r,v)$ such that $u \notin V$, we let $F_0(r,v) = \bot$. (2)  We rename the vertices $v \in V$ in $1,2, \cdots, N$
according to the order sampled and call this mapping $\varphi$. For $v \notin V$, we design $\varphi$ such that $\varphi(v)$ also assigns $v$ to some label in $\{N+1, \ldots, M\}$ in arbitrary fashion.

We then create oracle $F$: on query $(v,r), v \in [N]$, we map it to $v' = \varphi^{-1}(v)$ and query oracle $F_0$ on $(v', r)$ to get result $u'$ (or $\bot$) and output $u = \varphi(u')$ or $\bot$. We can observe that $F_0$ and $F$ are essentially the same oracle with only renaming on the vertices using oracle $\varphi, \varphi^{-1}$. If there is a quantum algorithm with access to oracle $F$, one can perfectly simulate its behavior and output with oracle access to $F_0$ and $\varphi, \varphi^{-1}$. Since $\varphi$ is only a relabeling of all vertices, it will not affect the performance of the algorithm. Therefore, from now on, we analyze the quantum algorithm $A$'s with respect to oracle $F_0$.
 
We then compute the fraction of changes on the output of oracle
$F_0$ from $F'$. These are 
all $(r,v)$ satisfies that $F'(r,v) \neq \bot, F_0(r,v) = \bot$. These correspond to either the case where vertex $v$ is not picked into set $V$ or the case where $v$'s neighbor $u = F'(r,v)$ on matching $r$ is not picked into set $V$ in the above oracle change procedure.
The expected fraction of output changes is: 
\begin{align*}
    & \Pr_{(r,v)}[F_0(r,v) = \bot,F'[r,v] \neq \bot] \\
= & \Pr_{(r,v)}[(v \in V', v\notin V) \vee (F'(r,v) \in V', F'(r,v) \notin V)] \leq \frac{2N^{0.9}}{N+\Theta(N^{0.9})} = \Theta \Big (\frac{1}{N^{1/10}} \Big)
\end{align*}

Let $T$ be the number of queries of $A$. We then have $\sum^{i = T}_{i, F_0(v,r) \neq F'(r,v)} W_{r,v}(\ket{\phi_i}) \leq \Theta(\frac{T}{N^{1/10}})$, where $W_{r,v}(\ket{\phi_i})$ is the squared amplitude of query weight of the $i$-th query on input 
$(r,v)$. Then by \Cref{thm:bbbv97_oraclechange}, we have that $\Vert \ket{\phi_{T+1}'} - \ket{\phi_{T+1}}  \Vert \leq \frac{T}{N^{1/20}}$ where $\ket{\phi_{T+1}'}$ and $\ket{\phi_{T+1}}$ are the algorithm's states in the world using oracle $F'$ and $F$ respectively. Since by assumption $T = o(N^{1/20})$, the above two worlds are indistinguishable; thus $A$ is an algorithm that can distinguish between the Yes and No instances of the subsampled distributions described above with advantage $\Omega(1)$. However,~\cite{ambainis2011quantum} also proved that any algorithm that does so must make at least $\Omega(N^{1/4}/\log N)$ queries, which is a contradiction. Thus there is no such algorithm $A$.

\end{proof}

\paragraph{Adversary Method.} 
Finally, we recall the famous Adversary method of Ambainis~\cite{ambainis2000quantum}. This will be used in our interactive game lower bound (\Cref{sec:interactive-game}).

\begin{theorem}[Adversary method]
\label{thm:adversary_method}
    Let $L \subseteq \{0, 1\}^*$
be a decision problem. Let
$X \subseteq L$ be a set of positive instances and $Y \subseteq L$ a set of negative instances. Let $R \subseteq X \times Y$ be a relation between instances of same size. Let be the values $m, m', l_{x,i}$ and $l_{y,i}'$ with $x, y \in \{0, 1\}^n$
 and $x \in X, y \in Y, i \in [n]$ such that:

 \begin{enumerate}
     \item for every $x \in X$ there are at least $m$ different $y \in Y$ such that $(x,y) \in R$.

     \item for every $y \in Y$ there are at least $m'$ different $x \in X$ such that $(x,y) \in R$.

     \item for every $x \in X$ and $i \in [n]$ there are at most $l_{x,i}$ different $y \in Y$ 
which differ from $x$ at entry $i$ and $(x,y) \in R$.

\item for every $y \in Y$ and $i \in [n]$ there are at most $l_{y,i }'$ different $x \in X$
which differ from $y$ at entry $i$ and $(x,y) \in R$.
\end{enumerate}
Then the quantum query complexity of computing $L$ with success probability $\epsilon$ is 
\[
    \Omega \Big ( \frac{m m'}{\max_{(x,r) \in R} \max_i l_{x,i} l_{y,i}'} \Big)~.
\]
\end{theorem}

\section{Min-entropy decompositions for permutations}
\label{sec:min-entropy-decompositions}

The quantum pseudorandomness conjectures discussed in the Introduction are centered around notions of ``dense'' random variables that are either uniform over a product alphabet $\Sigma^N$ or the set of permutations on $[N]$. Such dense random variables arise as follows: every distribution over strings or permutations with sufficiently high min-entropy can be approximately decomposed into a convex combination of distributions where a small subset of coordinates are fixed to a deterministic value, and the rest of the coordinates form a dense distribution. This decomposition lemma has been used in numerous settings in theoretical computer science, ranging from communication complexity lower bounds~\cite{goos2016rectangles,kothari2021approximating} to analyzing security in the presence of auxiliary information~\cite{unruh2007random,coretti2018random,guo2021unifying}.

We present a self-contained presentation of the definition of dense distributions over (tuples of) permutations, and the min-entropy decomposition lemma. In what follows, for a function $Z: \Sigma \to \Sigma$ and a subset $I \subseteq \Sigma$, we let $Z(I)$ denote the vector $\Sigma^I$ whose coordinates are the evaluations of $Z$ on the elements of $I$. We write $\mathrm{image}(Z(I)) = \{ a \in \Sigma : x(b)=a \text{ for some $b \in I$} \}$ to denote the set of evaluations.

We say that a random function $\rv{Z}: \Sigma \to \Sigma$ for some finite alphabet $\Sigma$ is a \emph{random permutation} if every instantiation of $\rv{Z}$ is a bijection between $\Sigma$ and itself. We let $\rv{Z}^{-1}$ denote the random function induced by sampling $Z \sim \rv{Z}$, and then taking the inverse $Z^{-1}$ (which is well-defined because $Z$ is a permutation).

\begin{definition}[Nice random permutations]
Let $\Sigma$ be a finite alphabet. 
    We say that a random permutation $\rv{Z}: \Sigma \to \Sigma$ is \emph{nice} if for all disjoint subsets $I,J \subseteq \Sigma$ and $z_J \in \supp(\rv{Z}(J))$ (the support of the marginal distribution of $\rv{Z}$ on coordinates $J$),
\[
    |\supp(\rv{Z}(J))| \cdot |\supp(\rv{Z}(I) \mid \rv{Z}(J) = z_J)| = |\supp(\rv{Z}(I \cup J))|
\]
where $\supp(\rv{Z}(I) \mid \rv{Z}(J) = z_J)$ denotes the support $\rv{Z}(I)$ conditioned on $\rv{Z}(J) = z_J$. We call this the \emph{chain rule condition} for nice random permutations. We call the following the \emph{entropy for $\rv{Z}$}:
\[
    h(I \mid J) := \log |\supp(\rv{Z}(I) \mid \rv{Z}(J) = z_J)|
\]
where $z_J \in \supp(\rv{Z}(J))$~.
\end{definition}

A canonical example of a nice random permutation is a uniformly random permutation $\rv{Z}: \Sigma \to \Sigma$; the associated entropy is
\[
    h(I \mid J) = \log \Big( (|\Sigma| - |J|)(|\Sigma| - |J| -1) \cdots (|\Sigma| - |J| - |I| +1) \Big)
\]
for all disjoint subsets $I,J \subseteq \Sigma$. Concatenations of uniformly random permutations are also nice permutations. Let $\rv{Z}_1:\Sigma_1 \to \Sigma_1,\ldots,\rv{Z}_R: \Sigma_R \to \Sigma_R$ be uniformly random permutations for some finite alphabets $\Sigma_1,\ldots,\Sigma_R$. Let $\Sigma = \bigcup_{r} \{r\} \times \Sigma_r$. Then the concatenation $\rv{Z} = \rv{Z}_1 \cdots \rv{Z}_R: \Sigma \to \Sigma$ is a random permutation where $\rv{Z}(r,x) = \rv{Z}_r(x)$. This satisfies the chain rule condition and its associated entropy $h$ is simply the following: for all disjoint $I,J \subseteq \Sigma$, letting $I = \bigcup_r \{r \} \times I_r$ and $J = \bigcup_r \{r \} \times J_r$, we have
\[
    h(I \mid J) = \sum_r h_r(I_r \mid J_r)
\]
where $h_r$ is the entropy of $\rv{Z}_r$.
    
\begin{definition}[Dense random permutations]
\label{def:dense}
Let $\rv{Z}: \Sigma \to \Sigma$ be a nice random permutation and let $h$ be its associated entropy. We say that a random permutation $\rv{X}: \Sigma \to \Sigma$ is \emph{$(k,\delta)$-dense with respect to $\rv{Z}$} if 
\begin{enumerate}
    \item $\supp(\rv{X}) \subseteq \supp(\rv{Z})$, and
    \item $\rv{X}$ is fixed in some subset $I \subseteq \Sigma$ of coordinates of size at most $k$ and for all $J \subseteq \Sigma \setminus I$,
\[
    \Hmin(\rv{X}(J)) \geq (1 - \delta) h(J \mid I)~.
\]
\end{enumerate}
If $\rv{X}$ is $(0,\delta)$-dense we say $\rv{X}$ is $\delta$-dense.
\end{definition}

The following is a simple but useful decomposition lemma for high min-entropy permutations.

\begin{lemma}[Decomposition of high min-entropy permutations]
\label{lem:high-min-entropy-decomposition}
Let $\rv{Z}: \Sigma \to \Sigma$ be a nice random permutation with associated entropy $h$. Let $\rv{X}: \Sigma \to \Sigma$ be a random permutation such that $\supp(\rv{X}) \subseteq \supp(\rv{Z})$ and $\Hmin(\rv{X}) \geq h(\Sigma) - d$. Then for all $0 < \delta,\eps < 1$, the following holds: $\rv{X}$ is $\eps$-close in total variation distance to a convex combination of random permutations $\{ \rv{Y}_z: \Sigma \to \Sigma \}_z$ where for each $z$, the random permutations $\rv{Y}_z$ and $\rv{Y}_z^{-1}$ are both $\Big (k,\delta \Big)$-dense with respect to $\rv{Z}$ and $k$ satisfies 
\[
    h(I) \leq \frac{1}{\delta} \Big(d + \log 1/\eps \Big)
\]
for all size-$k$ subsets $I \subseteq \Sigma$. %
\end{lemma}
We note that the following proof is essentially identical to the min-entropy decomposition lemma in~\cite[Appendix A]{coretti2018non}, except there is the additional guarantee that the inverses of each component $\rv{Y}_z$ are also dense\footnote{
\cite{coretti2018non} only prove that the ``forward'' permutations $\rv{Y}_z$ are dense; as mentioned in the Introduction, they implicitly also need that the inverse permutation distributions are also dense. However, the denseness of inverse permutation distributions is implied by their proof as well.}.
\begin{proof}
    We iteratively decompose $\rv{X}$. If $\rv{X}$ is already $\delta$-dense, then we are done. Otherwise, let $I \subseteq \Sigma$ denote the maximal-sized subset and a fixing $x_I \in \Sigma^I$ such that
    \[
        \Pr(\rv{X}(I) = x_I) \geq 2^{-(1 - \delta)h(I)}~.
    \]
    Here, ``maximal'' means that a coordinate $a \in \Sigma$ cannot be added to $I$ without the $\rv{X}(I \cup \{a\})$ being $\delta$-dense. Let $\rv{Y}_1$ denote the random variable $\rv{X}$ conditioned on $\rv{X}(I) = x_I$. Then $\rv{Y}_1$ is $(|I|,\delta)$-dense with respect to $\rv{Z}$ (note that $\supp(\rv{Y}_1) \subseteq \supp(\rv{X}) \subseteq \supp(\rv{Z})$). This is because if it wasn't, there exists a $J \subseteq \Sigma \setminus I$ and a fixing $y_J \in \Sigma^J$ such that 
    \[
        \Pr(\rv{Y}_1(J) = y_J) = \Pr(\rv{X}(J) = y_J \mid \rv{X}(I) = x_I) > 2^{-(1 - \delta)h(J \mid I)}~.
    \]
    On the other hand, due to the chain rule,
    \begin{align*}
        \Pr(\rv{X}(I \cup J) = x_I y_J) &= \Pr(\rv{X}(I) = x_I) \cdot \Pr(\rv{X}(J) = y_J \mid \rv{X}(I) = x_I)  \\
        & > 2^{-(1 - \delta)h(I)} \cdot 2^{-(1 - \delta)h(J \mid I )}  \\
        &= 2^{-(1 - \delta)h(I \cup J)}
    \end{align*}
    violating the maximality of $I$.

    Observe that the inverse permutation $\rv{Y}^{-1}$ is the random variable $\rv{X}^{-1}$ conditioned on $\rv{X}^{-1}(x_I) = I$ (by which we mean that $\rv{X}^{-1}(a) = b$ when $x_I(b) = a$ for all $b \in I$, $a \in \Sigma$). We now argue that  $\rv{Y}^{-1}$ is also $(|I|,\delta)$-dense with respect to $\rv{Z}$. Suppose not; then there exists a $J \subseteq \Sigma \setminus \mathrm{image}(x_I)$ \footnote{Observe that $J$ is a set of image points of $\rv{X}$/ $\rv{Y}$.} and a fixing $y_J \in \Sigma^J$ such that 
    \[
        \Pr(\rv{Y}_1^{-1}(J) = y_J) = \Pr(\rv{X}^{-1}(J) = y_J \mid \rv{X}^{-1}(x_I) = I) > 2^{-(1 - \delta)h(J \mid \mathrm{image}(x_I))} = 2^{-(1 - \delta)h(J \mid I)}
    \]
    where the last equality uses that $|\mathrm{image}(x_I)| = |I|$ and that $h$ only depends on the cardinality of its arguments. But once again, due to the chain rule,
    \begin{align*}
        \Pr(\rv{X}^{-1}(x_I) = I \wedge \rv{X}^{-1}(J) = y_J) &= \Pr(\rv{X}^{-1}(x_I) = I) \cdot \Pr(\rv{X}^{-1}(J) = y_J \mid \rv{X}^{-1}(x_I) = I)  \\
        &= \Pr(\rv{X}(I) = x_I) \cdot \Pr(\rv{X}^{-1}(J) = y_J \mid \rv{X}^{-1}(x_I) = I) \\
        & > 2^{-(1 - \delta)h(I)} \cdot 2^{-(1 - \delta)h(J \mid I )}  \\
        &= 2^{-(1 - \delta)h(I \cup J)}~.
    \end{align*}
    This implies
    \[
    \Pr(\rv{X}(I) = x_I \wedge \rv{X}(y_J)= J) > 2^{-(1 - \delta)h(I \cup J)} = 2^{-(1 - \delta)h(I \cup \mathrm{image}(y_J))}
    \]
    violating the maximality of $I$. Therefore $\rv{Y}_1$ must be $(|I|,\delta)$-dense as desired.

    Next, we bound $h(I)$. Let $\overline{I} = \Sigma \setminus I$ and note that
    \begin{align*}
        \Pr(\rv{X}(I) = x_I) &= \sum_{x_{\overline{I}} \in \supp(\rv{X}_{\overline{I}})} \Pr(\rv{X} = x_I \cup x_{\overline{I}}) \\
        &\leq |\supp(\rv{X}(\overline{I}) \mid \rv{X}(I) = x_I)| \cdot 2^{-h(\Sigma) + d} \\
        &= 2^{h(\overline{I} \mid I)} 2^{- h(\Sigma) + d} \\
        &= 2^{- h(I) + d} \tag{chain rule}
    \end{align*}
    But on the other hand $\Pr(\rv{X}_I = x_I) \geq 2^{-(1 - \delta)h(I)}$ by assumption, so this means that $h(I) \leq d/\delta$. 

Now let $\rv{X}_2$ denote $\rv{X}$ conditioned on $\rv{X}(I) \neq x_I$ and recursively decompose $\rv{X}_2$ to get $\rv{Y}_2$ and $\rv{Y}_3$, etc. as long as $\Pr(\rv{X} \in \supp(\rv{X}^{(j)})) \geq \eps$. Observe that at any point in the decomposition process, $\Hmin(\rv{X}^{(j)}) \geq h(\Sigma) - (d + \log 1/\eps)$ because
\begin{align*}
    \Pr(\rv{X}^{(j)} = x) &= \Pr(\rv{X} = x \mid \rv{X} \in \supp(\rv{X}^{(j)})) \\
                     &\leq \frac{\Pr(\rv{X} = x)}{\Pr(\rv{X} \in \supp(\rv{X}^{(j)}))} \\
                     &\leq \frac{1}{\eps} \cdot 2^{-h(\Sigma) + d} \\
                     &= 2^{-h(\Sigma) + d + \log 1/\eps}~.
\end{align*}
\end{proof}

\begin{claim}
\label{clm:random-permutation-entropy-vs-set}
    Let $h(I)$ denote the entropy associated with a uniformly random permutation on a set $\Sigma$. Then $h(I) \leq k$ implies that $|I| \leq k +1$.
\end{claim}
\begin{proof}
    Suppose that $|I| \geq 1$ (otherwise claim is trivial). It also must be that $|I| \leq |\Sigma|$ (otherwise it's impossible to have a permutation). Then observe that
    \[
        h(I) = \log |\Sigma| \cdots (|\Sigma|- |I| +1) \geq (|I| - 1) \cdot \log (|\Sigma|- |I| + 2) \geq |I| - 1~.
    \]
    Therefore $|I| \leq k+1$.
\end{proof}

\section{Quantum pseudorandomness conjectures}
\label{sec:pseudorandomness}

\subsection{A link between \Cref{conj:simplified} and the Aaronson-Ambainis conjecture}

The pseudorandomness conjecture in the random oracle model (\Cref{conj:simplified}) is very similar to Conjecture 4 from~\cite{guo2021unifying}, which was shown to imply that quantum query algorithms can be efficiently simulated by classical deterministic algorithms over the uniform distribution (this is the main quantum computing consequence of the Aaronson-Ambainis conjecture~\cite{aaronson2009need}). We present a self-contained proof that \Cref{conj:simplified} also implies a particular reformulation of the conjecture about classical simulability of quantum algorithms (i.e. Conjecture 3 of~\cite{guo2021unifying}).

\begin{proposition}
\label{prop:conj-uniform-implies-aa}
    Assume that \Cref{conj:simplified} is true. Let $A$ denote a $T$-query quantum algorithm that queries a uniformly-random $N$-bit string $X \in \{0,1\}^N$. Let $\mathrm{Var}(A)$ denote 
    \[
        \mathrm{Var}(A) := \E_{X \sim \rv{U}} \Big( \Pr[A^X = 1]^2 \Big) -  \Big (\E_{X \sim \rv{U}} \Pr[A^X = 1] \Big)^2
    \]
    where the expectation is over the uniform random variable $\rv{U}$ on $\{0,1\}^N$, and the probabilities are over the randomness of the algorithm $A$ only.
    Then there exists a subset $I \subseteq [N]$ of size at most $\poly(T/\mathrm{Var}(A))$ and a string $W \in \{0,1\}^I$ such that
    \[
        \Big |  \Pr_{Z \sim \rv{Z}} [A^Z = 1] -  \Pr_{X \sim \rv{U}}  [A^X = 1] \Big| \geq \Omega(\sqrt{\mathrm{Var}(A)})
    \]
    where $\rv{Z}$ is a uniformly random string in $\{0,1\}^N$ conditioned on $\rv{Z}_I = W$, and the probabilities are both over the randomness of sampling the oracle as well as the randomness of the algorithm. 
\end{proposition}

\begin{proof}
We largely follow the proof of~\cite[Theorem 4]{guo2021unifying}. 
For notational convenience, for all $Q \in \{0,1\}^N$, define $A^{Q} := \Pr[A^Q = 1]$ where the probability is over the randomness of the algorithm. 
Let  $\mu := \E_{X \sim \rv{U}} [A^X]$ %
and let $\sigma := \sqrt{\mathrm{Var}(A)}$. From the definition of variance we can write
\begin{align*}
    \sigma^2 = \E_{X\sim \rv{U}} [(A^X - \mu)^2] \leq  \Pr_{X \sim \rv{U}} (|A^X -\mu| \geq \sigma/2) + (\sigma/2)^2~.
\end{align*}
Rearranging, we get $\Pr_{X\sim \rv{U}}(|A^X - \mu|\geq \sigma/2)\geq 3\sigma^2/4$. Thus at least one of $\Pr_{X\sim \rv{U}}[A^X \geq \mu+  \sigma/2]$ or $\Pr_{X\sim \rv{U}}[A^X \leq \mu -  \sigma/2]$ must be at least $3\sigma^2/8$. Suppose (without loss of generality) that $\Pr_{X\sim \rv{U}}[A^X \geq \mu+  \sigma/2] \geq 3\sigma^2/8$.

Let $\cal{V}$ denote the set of $X$'s such that $A^X \geq \mu + \sigma/2$, and let $\rv{V}$ denote a uniformly random element of $\cal{V}$. 
Since $|\cal{V}| \geq (3\sigma^2/8) \cdot 2^N$ we have
\[
\Hmin(\rv{V})\geq N- \log \frac{8}{3\sigma^2}~.
\] 
The min-entropy decomposition lemma in~\cite[Appendix A]{coretti2018random} then implies that for all $0 < \delta,\eps < 1$, the random variable $\rv{V}$ is $\epsilon$-close to a convex combination of random variables $\{\rv{Y}^{(z)}\}_z$  that are  $(k, \delta)$-dense for $k=\delta^{-1} \Big( \log 8/3\sigma^2+\log(1/\epsilon)\Big)$. This implies that $\Big | \E_{z} \Pr_{Y \sim \rv{Y}^{(z)}} [A^Y = 1]- \Pr_{V \sim \rv{V}} [A^V = 1] \Big | \leq \epsilon$ and therefore 
\[
\E_{z} \Pr_{Y \sim \rv{Y}^{(z)}} [A^Y = 1] \geq \mu + \sigma/2 - \eps~.
\]
By averaging there exists a $z$ such that $\Pr_{Y \sim \rv{Y}^{(z)}} [A^Y = 1]  \geq \mu + \sigma/2 - \eps$. For notational convenience we drop the $z$ and simply write $\rv{Y} = \rv{Y}^{(z)}$.

Since $\rv{Y}$ is $(k,\delta)$-dense, there is a subset $I$ of at most $k$ coordinates such that $\rv{Y}$ is fixed to some string $W \in \{0,1\}^I$ and the rest of the string $\rv{Y}_{\overline{I}}$ is a $\delta$-dense random variable over $\{0,1\}^{N - |I|}$. The algorithm $A$ making queries to $X$ sampled from $\rv{Y}$ can be viewed as making queries to $\rv{Y}_{\overline{I}}$ only, since the coordinates $I$ are fixed. By \Cref{conj:simplified}, letting $\rv{Z}$ be uniform over $\{0,1\}^N$ conditioned on $\rv{Z}_I = W$, we have that
\[
    \Big | \Pr_{Z \sim \rv{Z}} [A^Z = 1] - \Pr_{Y \sim \rv{Y}} [A^Y = 1] \Big| \leq \poly(T) \cdot \poly(\delta)~.
\]
Therefore 
\[
    \Pr_{Z \sim \rv{Z}} [A^Z = 1] \geq \mu + \sigma/2 - \eps - \poly(T) \cdot \poly(\delta)~.
\]
By setting $\eps = \sigma/4$ and $\delta = \poly(\sigma/T)$ for some universal polynomials $\poly(\cdot)$ from \Cref{conj:simplified}, we get that this is at least $\Omega(\sigma)$. Observing that $k = \poly(T/\mathrm{Var}(A))$, we conclude the proof.

\end{proof}

\subsection{A weaker version of \Cref{conj:simplified}}

Given that the Aaronson-Ambainis conjecture has been open for a while it appears that the quantum pseudorandomness conjectures are also quite difficult to prove. However we also give some evidence that \Cref{conj:simplified} is true by proving a weaker version of it.

This is analogous to how Aaronson and Ambainis proved a weaker version of their eponymous conjecture via a Fourier-analytic result of Dinur, et al.~\cite{dinur2006fourier}. Indeed, we also use the same Fourier-analytic tool in our proof. 

We first recall a few basic concepts from analysis of Boolean functions and entropy before delving into the proof. 

\paragraph{Fourier analysis on the boolean hypercube.} We define inner product between two functions $f,g : \{ \pm 1\}^n \to \R$ to be $\ip{f,g} := \E_x f(x)g(x)$ where the expectation is over uniformly random $x \in \{\pm 1\}^n$. We define $p$-norms to be $\| f \|_p = (\E_x |f(x)|^p)^{1/p}$. We observe that the inner product satisfies H\"{o}lder's inequality: for $p,q > 0$ satisfying $\frac{1}{p} + \frac{1}{q} = 1$, we have $|\ip{f,g} | \leq \| f \|_p \, \| g \|_q$.

\begin{lemma}[$(2,q)$-hypercontractivity; Chapter 9 of~\cite{o2014analysis}]
\label{lem:hypercontractivity}
Let $f:\{\pm 1\}^n \to \R$ be a degree-$d$ function. Then for all integers $q > 2$ we have $\| f \|_q \leq \sqrt{q - 1}^d \| f \|_2$.
\end{lemma}

We say that a function $f: \{ \pm 1\}^n \to \R$ has a \emph{$(k,d,\eps)$-junta approximation} if there exists a degree-$d$, $k$-junta function $g: \{ \pm 1 \}^n \to \R$ such that $\| f - g \|_2 \leq \eps$.

\begin{theorem}[Dinur, Friedgut, Kindler, O'Donnell~\cite{dinur2006fourier}]
\label{thm:dfko}
Let $\eps > 0$. All degree-$d$ polynomials $f:\{\pm 1\}^n \to [-1,1]$ have a $(2^{O(d)}/\eps^2,d,\eps)$-junta approximation.
\end{theorem}

\paragraph{Densities and entropy.} A function $\mu:X \to \R_{\geq 0}$ is a \emph{density} if $|X|^{-1} \sum_{x \in X} \mu(x) = 1$. We say that a density $\mu$ has \emph{min-entropy deficit} $t$ if $\mu(x) \leq 2^t$ for all $x \in X$. If $X = \{\pm 1\}^n$ and $S \subseteq [n]$, then we say that the marginal of $\mu$ on the subset $S$ is the density $\mu_S$ where for all $x_S \in \{ \pm 1\}^S$
\[
    \mu_S(x_S) = \E_{x_{-S}} \mu(x_S,x_{-S})
\]
where the expectation is uniform over bits outside of $S$. We say that a density is \emph{$\delta$-dense} if for all subsets $S \subseteq [n]$, the marginal $\mu_S$ has min-entropy deficit at most $\delta |S|$. \\
\medskip

\Cref{lem:weaker} is a consequence of the following lemma about low-degree polynomials. 
\begin{lemma}
\label{lem:dfko-dense}
Let $f: \{ \pm 1\}^n \to [-1,1]$ be a degree-$d$ polynomial, let $0 < \eps \leq 1$, and let $\mu: \{\pm 1\}^n \to \R_{\geq 0}$ be a $\delta$-dense density where $\delta \leq \eps^2 2^{-O(d)}$. Then 
    \[
        |\ip{f, \mu - 1} | \leq \sqrt{\delta} 2^{O(d)}/\eps + O \Big( (\delta n)^{d/2} \, \eps \Big)~.
    \]
\end{lemma}
\begin{proof}
By \Cref{thm:dfko}, $f$ has a $(2^{O(d)}/\eps^2,d,\eps)$-approximation $\tilde f$. For convenience let $k = 2^{O(d)}/\eps^2$. Then
    \[
        | \ip{f, \mu - 1} | \leq \abs{\ip{ \tilde{f}, \mu - 1} } + \abs{\ip{ f - \tilde{f}, \mu - 1} }~.
    \]
    We bound each term separately. Let $S \subseteq [n]$ denote the $k$ variables that $\tilde{f}$ is defined on. First, since $\mu$ is $\delta$-dense, the marginal $\mu_S$ has min-entropy deficit at most $\delta k$. This implies that $\mu_S$ is $O(\delta k)$-close in total variation distance to the uniform density (i.e. the constant $1$ function). To see this, observe that the total variation distance can be expressed as
    \[
        2^{-k} \sum_{y \in \{ \pm 1\}^S} | \mu_S(y) - 1| = \frac{2}{2^k} \sum_{y : \mu_S(y) > 1} |\mu_S(y) - 1| \leq \frac{2}{2^k} \sum_{y \in \{\pm 1\}^S} |2^{\delta k} - 1| = O(\delta k)~.
    \]
    This follows since $\delta k \leq O(1)$ (i.e., it does not grow with $n$), and then $2^{\delta k} - 1 = O(\delta k)$ by Taylor series approximation.
    Therefore 
    \[
        \abs{\ip{ \tilde{f}, \mu - 1}} = \abs{\ip{\tilde{f}, \mu_S - 1}} \leq \| \tilde{f} \|_2 \, \| \mu_S - 1 \|_2
    \]
    by Cauchy-Schwarz. Since $\tilde{f}$ is $\eps$-close to $f$, which has bounded $2$-norm, this is at most
    \begin{align*}
        (1 + \eps) \sqrt{ \E_y (\mu_S(y) - 1)^2} %
        &\leq (1 + \eps) \sqrt{ \E_y | \mu_S(y) - 1| \cdot (2^{\delta k} + 1) } \leq O(\sqrt{\delta k})
    \end{align*}
    where we used the bound on the total variation distance between $\mu_S$ and uniform, and that $2^{\delta k} = O(1)$. 
    For the second term, we use H\"{o}lder's inequality and hypercontractivity (\Cref{lem:hypercontractivity}):
    \begin{align*}
        \abs{\ip{ f - \tilde{f}, \mu - 1} } &\leq \abs{\ip{ f - \tilde{f}, \mu} } + \abs{ \ip{f - \tilde{f},1}} \\
        &\leq \| f - \tilde{f} \|_p \, \| \mu \|_q + \| f - \tilde{f} \|_2 \\
    &\leq \Big ( (p-1)^{d/2}\| \mu \|_q +1 \Big )\, \| f - \tilde{f} \|_2 ~.
    \end{align*}
    Set $p = \delta n + 1$ and $q = \frac{p}{p-1}$. Then
    \[
        \| \mu \|_q^q = \E_x |\mu(x)| \cdot |\mu(x)|^{1/(p-1)} \leq \E_x |\mu(x)| \cdot 2^{\delta n/(p-1)} \leq 2
    \]
    where we used that the min-entropy deficit of $\mu$ is at most $\delta n$.
    Therefore 
    \[
        \abs{\ip{ f - \tilde{f}, \mu - 1} } \leq \Big ( (\delta n)^{d/2} \, 2^{\frac{p-1}{p}}~ + 1 \Big ) \, \eps \,  \leq O \Big( (\delta n)^{d/2} \, \eps \Big)~.
    \]
    Putting everything together, we get that
    \[
        | \ip{f, \mu - 1} | \leq O \Big( \sqrt{\delta k} +  (\delta n)^{d/2} \eps \Big)~.
    \]
    Plugging our choice of $k$, we get the desired statement.
\end{proof}

To conclude the proof of \Cref{lem:weaker}, we observe that the acceptance probability of a $T$-query quantum algorithm $A$ that has query access to a string $X \in \{0,1\}^N$ can be expressed as $f(X)$ where $f: \{0,1\}^N \to [-1,1]$ is a degree-$2T$ polynomial; this follows from the polynomial method of~\cite{beals2001quantum}. By making the switch from $\{0,1\}$ to $\{\pm 1\}$ variables we can then apply \Cref{lem:dfko-dense} and observe that
\[
|\ip{f, \mu - 1} | = \Big | \E_{X \sim \rv{X}} \Pr[A^X = 1] - \E_{Y \sim \rv{U}} \Pr[A^Y = 1] \Big |
\]
to obtain the desired conclusion.

\subsection{Towards proving special cases of \Cref{conj:pseudorandomness}}

\paragraph{Proving the conjecture for one permutation.} A natural question is whether we can deduce the quantum pseudorandomness conjecture (\Cref{conj:pseudorandomness}) from a weaker statement that only refers to a \emph{single} permutation $P$, rather than multiple permutations:

\begin{conjecture}[Quantum pseudorandomness conjecture for a single permutation]
\label{conj:pseudorandomness-single}
Let $A$ be a quantum query algorithm that makes $T$ queries to a permutation $P:[N] \to [N]$. Let $\rv{P}$ be a random permutation of $[N]$ such that both $\rv{P}$ and the inverse permutation $\rv{P}^{-1}$ are $\delta$-dense. Then
\[
    \Big | \Pr_{P \sim \rv{Z}}[A^P = 1] - \Pr_{P \sim \rv{P}}[A^P = 1] \Big| \leq \poly(T,\log N) \cdot \poly(\delta)
\]
where $\rv{Z}$ denotes a uniformly random permutation on $[N]$.
\end{conjecture}

A first strategy would be to try to do some sort of hybrid argument, e.g., replace a dense random variable $\rv{P} = (\rv{P}_1,\ldots,\rv{P}_R)$ with uniformly random permutations one component at a time, invoking \Cref{conj:pseudorandomness-single} each time. One runs into an issue when trying to carry out this approach: to switch even the first permutation $\rv{P}_1$ to the uniform distribution $\rv{Z}_1$ over $S_{N_1}$ would require arguing that conditioned on a sample $(P_2,\ldots,P_R)$ from the marginal distribution $(\rv{P}_2,\ldots,\rv{P}_R)$, the distribution of the first component $\rv{P}_1$ is still $\delta'$-dense for some controlled $\delta'$. 

Unfortunately, this is not true in general. Consider the following random variable $\rv{P} = (\rv{P}_1,\rv{P}_2)$ over $S_N \times S_N$: the random variable $\rv{P}_2$ is uniform over $S_N$, and $\rv{P}_1$ is uniform conditioned on $\rv{P}_1(m) = m$ where $m$ is the length of the longest cycle of $\rv{P}_2$. The min-entropy of $\rv{P}$ is $2\log N! - \log N$, and thus is $\delta$-dense for $\delta = \frac{\log N}{2 \log N!}$. On the other hand, fixing a suffix $P_2$ will fix a coordinate of $\rv{P}_1$, which is no longer dense. %

We leave it as an interesting question to explore whether \Cref{conj:pseudorandomness-single} can be ``boosted'' to imply \Cref{conj:pseudorandomness}.

\paragraph{Proving the conjecture for few queries.}
Another interesting question is whether it is possible to prove a permutation analogue of \Cref{lem:weaker}. The proof of this relied heavily on tools from boolean Fourier analysis, which works for well for studying distributions over product spaces. However, the distribution of random permutations is highly non-product: knowing that a random permutation maps $\pi(i) = j$ means that $j$ cannot appear in any other coordinate of $\pi$ (when viewed as a string in $[N]^N$). 

While the theory of Fourier analysis of boolean functions has been richly developed (see, e.g.,~\cite{o2014analysis}), the theory of harmonic analysis of functions on non-product spaces such as the symmetric group is nascent. The recent paper of Filmus, Kindler, Lifshitz, and Minzer~\cite{filmus2024hypercontractivity} (and of Keevash and Lifshitz~\cite{keevash2023sharp}) proved hypercontractive inequalities for functions on the symmetric group. Exploring the utility of these results in harmonic analysis to obtain additional evidence for \Cref{conj:pseudorandomness} is something we leave for future work.

\subsection{A helper lemma}

In our $\QCMA$ lower bound we deal with distributions over permutations where the permutations are fixed in a few coordinates, and are dense everywhere outside of those fixings. Such distributions do not exactly fulfill the requirements of \Cref{conj:pseudorandomness}; technically the conjecture only handles random permutations that are dense \emph{everywhere} and do not have any fixed coordinates. We state a variant of \Cref{conj:pseudorandomness} that handles the broader class of permutation distributions, and argue that it is implied from \Cref{conj:pseudorandomness}. 

Consider a random variable $\rv{B} = (\rv{B}_1,\ldots,\rv{B}_R)$ where each $\rv{B}_r$ is a permutation on $[N_r]$. We can equivalently view it as a random permutation $\rv{B}$ over the alphabet $\Sigma = \bigcup_r \{r\} \times [N_r]$ such that for each $r$, $\rv{B}(r,\cdot)$ permutes $[N_r]$. \Cref{def:dense} specifies the notion of $(k,\delta)$-denseness for such random permutations, with respect to the random permutation $\rv{Z}_B$ that is uniform over $S_{N_1} \times \cdots \times S_{N_R}$ conditioned on the same $k$ fixed coordinates as in $\rv{B}$. 

\begin{lemma}
\label{lem:conjhelper}
If \Cref{conj:pseudorandomness} is true, then the following is true. Let $A$ be a quantum query algorithm that makes $T$ queries to a collection of permutations $(B_1,\ldots,B_R)$ where each $B_r$ permutes $[N_r]$ for some integer $N_r$. Let $\rv{B}$ denote a random tuple of permutations from $S_{N_1} \times \cdots \times S_{N_R}$, such that both $\rv{B}$ and $\rv{B}^{-1}$ are both $(k,\delta)$-dense. Then
\[
    \Big | \Pr_{Z \sim \rv{Z}_B}[A^Z = 1] - \Pr_{B \sim \rv{B}}[A^B = 1] \Big| \leq \poly(T, \log N_1,\ldots,\log N_R) \cdot \poly(\delta)
\]
where $\rv{Z}_B$ denotes a tuple of uniformly random permutations sampled from $S_{N_1} \times \cdots \times S_{N_R}$ \emph{conditioned} on the same $k$ fixed coordinates as $\rv{B}$.

\end{lemma}
\begin{proof}
Let $\rv{Z}_B,\rv{B}$ be random variables as specified in the lemma statement. Since $\rv{B}$ is $(k,\delta)$-dense, for each $r \in [R]$ there exists sets $I_r,J_r \subseteq [N_r]$ such that $\rv{B}_r$ is \emph{not} fixed on $I_r$ and $\rv{B}_r(I_r) = J_r$. In particular, $|I_r| = |J_r|$ for all $r$ and $\sum_r |I_r| \geq \sum_r N_r - k$. For all $r \in [R]$, let $\pi_r: J_r \to I_r$ be some fixed bijection. 

For all tuples $(B_1,\ldots,B_R)$, define the permutations $P_r: I_r \to I_r$ as follows:
\[
P_r(x)= \pi_r(B_r(x))~.
\]
Let $\rv{P}_r$ denote the random variable induced by $\rv{B}_r$, and let $\rv{P}$ denote the product random variable $\rv{P}_1 \times \dots \times \rv{P}_R$. Let $\rv{Z}$ denote a tuple of uniformly random permutations $(P_1,\ldots,P_R)$ sampled from $S_{I_1}\times \dots \times S_{I_R}$.  
For all subsets $T \subseteq \bigcup_r \{r\} \times I_r$, we have 
\[
\Hmin(\rv{P}(T))= \Hmin(\pi(\rv{B}(T))) = \Hmin(\rv{B}(T)) \geq (1 - \delta)h(T)
\]
where $\pi(a_1,\ldots,a_R) = (\pi_1(a_1),\ldots,\pi_R(a_R))$ and $h$ is the entropy associated with $\rv{Z}$. The second equality is because the min-entropy is invariant under renaming under the bijections $\pi_r$, and the last inequality is because $\rv{B}$ is $\delta$-dense outside $\bigcup_r \{r \} \times I_r$. Thus if $\rv{B}$ is $(k,\delta)$-dense  with respect to $\rv{Z}_B$, then $\rv{P}$ is $\delta$-dense with respect to $\rv{Z}$. Observe that $P_r^{-1}(y) = B_r^{-1}(\pi_r^{-1}(y))$. Then, for all subsets $S\subseteq \bigcup_r \{r \} \times I_r$,
\[
\Hmin(\rv{P}^{-1}(S))= \Hmin(\rv{B}^{-1}(\pi^{-1}(S)))\geq (1 - \delta)h(\pi^{-1}(S))=(1 - \delta)h(S).
\]
The last equality follows because $h$ only depends on the cardinality of its input.

Define an algorithm $A'$ that makes $T$ queries to the collection of permutations $(P_1, \dots P_R)$ and simulates $A$ as follows:
\begin{itemize}
    \item Whenever $A$ queries $B_r(x)$, check if $x\in [N_r]\setminus I_r$. If yes, return $B_r(x)$. Otherwise, query $y \leftarrow P_r(x)$ and return $\pi_r^{-1}(y)$ to $A$.
    \item  Whenever $A$ queries $B^{-1}_r(y)$, check if $y\in [N_r]\setminus J_r$. If yes, return the value $x$ such that $B_r(x)=y$. Otherwise, query $q \leftarrow \pi_r(y)$ and return $P_r^{-1}(q)$ to $A$.
    \item  Output $A$'s output.
\end{itemize}

  We can conclude that $\Pr_{P \sim \rv{Z}}[A'^{P}=1]=\Pr_{B \sim \rv{Z_B}}[A^{B}=1]$, and  $\Pr_{P \sim \rv{P}}[A'^{P}=1]=\Pr_{B \sim \rv{B}}[A^{B}=1]$.  \Cref{conj:pseudorandomness} implies that
 \[
    \Big | \Pr_{B \sim \rv{Z}_B}[A^B = 1] - \Pr_{B \sim \rv{B}}[A^B = 1] \Big| \leq \poly(T, \log N_1,\ldots,\log N_R) \cdot \poly(\delta)
\]
as desired.
\end{proof}

\section{$\QMA$ upper bound}
\label{sec:qma_upper_bound}

We give a self-contained proof of the $\QMA$ upper bound for the $\Components$ problem. This is essentially the same $\QMA$ upper bound given by Natarajan and Nirkhe~\cite{natarajan2023distribution}, with some minor simplifications (for example, we obtain perfect completeness whereas their $\QMA$ verifier has imperfect completeness).

\begin{longfbox}[breakable=false, padding=1em, margin-top=1em, margin-bottom=1em]
\begin{algorithm}\label{qma_verifier_algo}{ $\QMA$ verifier for the $\Components$ problem}
\end{algorithm}
\noindent Given oracle access to a graph oracle $F$, and an $n$-qubit witness $\reg{\psi}$ in register $\reg{X}$, perform one of the following tests at random:
    \begin{enumerate}
        \item \textbf{Balancedness test.} Apply $H^{\otimes n}$ to $\ket{\psi}$ and check if the result is $\ket{0^n}$. If so, reject. Otherwise, accept.

        \item \textbf{Invariance test.} 
        \begin{enumerate}
            \item Sample a random $r \in [R]$.
            \item Prepare a control qubit $\reg{C}$ in the state $\ket{+}$, and a register $\reg{R}$ in the state $\ket{r}$.
            \item Controlled on $\reg{C}$ being in the state $\ket{1}$, $\reg{R}$ in the state $\ket{r}$, apply the permutation $F(r,\cdot)$ to the register $\reg{X}$.
            \item Apply $H$ to $\reg{C}$ and check if the result is $\ket{0}$. If so, accept. Otherwise, reject.
        \end{enumerate}
    \end{enumerate}
\end{longfbox}

The meaning of ``apply the permutation'' $F(r,\cdot)$ to register $\reg{X}$ is that the verifier applies the permutation \emph{in-place}, meaning that $\ket{x}$ gets mapped to $\ket{F(r,x)}$. This can be done by attaching an ancilla register $\ket{0}$, computing the ``forwards'' oracle to obtain $\ket{x} \ket{F_r(x)}$, uncomputing the first register using the ``inverse'' oracle to obtain $\ket{0} \ket{F_r(x)}$, and then swapping the registers.

\begin{lemma}
\label{lem:qma_upper_bound_correctness}
    If $F \in \Lyes$, then there exists a witness $\ket{\psi}$ such that the $\QMA$ verifier $V$ accepts with probability at $1$. %
\end{lemma}
\begin{proof}
    Since $F \in \Lyes$, there exists a unique partition $(S,\overline{S})$ of $[N]$ such that for all $r \in R$, the oracle $F$ simply permutes the vertices within $S$ and $\overline{S}$ respectively. Define
    \[
        \ket{\psi} := \sqrt{\frac{1}{N}}\sum_{x \in S} \ket{x} - \sqrt{\frac{1}{N}} \sum_{x \notin S} \ket{x}~.
    \]
    The probability that witness state $\ket{\psi}$ is rejected in the Balancedness test is exactly equal to the squared overlap between $\ket{\psi}$ and the uniform superposition which is $0$ because $|S| = N/2$. Thus the verifier always accepts in this subtest.
    
    We now compute its success probability in the Invariance test. 
    After registers $\reg{C}$ and $\reg{R}$ are prepared, the state of the algorithm is as follows:
    \[
        \ket{+}_{\reg{C}} \otimes \ket{r}_{\reg{R}} \otimes \ket{\psi}_{\reg{X}}
    \]
    Let $\pi$ denote the permutation $F(r,\cdot)$. 
    After applying the oracle
    we have (omitting mention of the register $\reg{R}$)
    \[
        \frac{1}{\sqrt{2}} \ket{0}_{\reg{C}} \otimes \ket{\psi}_{\reg{A}} \otimes  + \frac{1}{\sqrt{2}}  \ket{1}_{\reg{C}} \otimes \pi  \ket{\psi}_{\reg{X}}~.
    \]
    Hadamarding the control qubit we get
    \[
        \frac{1}{2} \ket{0}_{\reg{C}} \otimes \Big( \ket{\psi} + \pi \ket{\psi}\Big) + \frac{1}{2} \ket{1}_{\reg{C}} \otimes \Big( \ket{\psi} - \pi \ket{\psi}\Big)~.
    \]
    The probability of rejecting (conditioned on $r$) is
    \[
        \frac{1}{4} \Big \| (I - \pi) \ket{\psi} \Big \|^2 = \frac{1}{4} \Big \| \sqrt{\frac{1}{N}}\sum_{x \in S} \Big( \ket{x} - \ket{\pi(x)} \Big) + \sqrt{\frac{1}{N}}\sum_{x \notin S} \Big( \ket{x} - \ket{\pi(x)} \Big) \Big \|^2~.
    \]
    By the triangle inequality, this is at most
    \[
        \frac{|S|}{2N}  \Big \| \frac{1}{\sqrt{|S|}} \sum_{x \in S} \Big( \ket{x} - \ket{\pi(x)} \Big) \Big \|^2 + \frac{N - |S|}{2N} \Big \| \frac{1}{\sqrt{N - |S|}} 
 \sum_{x \notin S} \Big( \ket{x} - \ket{\pi(x)} \Big) \Big \|^2~.
    \]
    Since $F$ is a $\Lyes$ instance, we have that $x \in S$ if and only if $\pi(x) \in S$, and thus
    \begin{gather*}
        \Big \| \frac{1}{\sqrt{|S|}} \sum_{x \in S} \Big( \ket{x} - \ket{\pi(x)} \Big) \Big \|^2 = \Big \| \frac{1}{\sqrt{N - |S|}} 
 \sum_{x \notin S} \Big( \ket{x} - \ket{\pi(x)} \Big) \Big \|^2 = 0~.
    \end{gather*}
    Since this holds for every $r$, the probability of rejection in the Invariance test is $0$. The overall acceptance probability is therefore $1$. 
\end{proof}

\begin{lemma}
    If $F \in \Lno$, then for all witnesses $\ket{\psi}$ the verifier $V$ accepts with probability at most $1 - \Delta$, where $\Delta$ is the spectral gap of the  normalized Laplacian of the graph associated with $F$.
\end{lemma}
\begin{proof}
    Let $\ket{\psi} = \sum_x \alpha_x \ket{x}$. Suppose that the verifier $V$ accepts with probability at least $1 - \eps$ in the Invariance test. From the above, the rejection probability is exactly equal to
    \[
        \frac{1}{4} \E_r \Big \| \ket{\psi} - \pi_r \ket{\psi} \Big\|^2 \leq \eps~.
    \]
    Expanding, this implies
    \[
        \E_r \bra{\psi} ( I - \pi_r^\dagger)(I - \pi_r) \ket{\psi} = \E_r \bra{\psi}(2 I - \pi_r - \pi_r^\dagger) \ket{\psi} = \bra{\psi} (2I - \E_r (\pi_r + \pi_r^\dagger) \ket{\psi} \leq 4\eps.
    \]
    Notice that for all $r$, the $(x,y)$ element of the matrix $\pi_r + \pi_r^\dagger$ is $1$ if and only if either $\pi_r(x) = y$ or $\pi_r(y) = x$. Thus $\E_r (\pi_r + \pi_r^\dagger)$ is proportional to the normalized adjacency matrix $\frac{1}{2R} A$ of the graph $G$ corresponding to $F$. Letting $\cal{L} = I - \frac{1}{2R} A$ denote the normalized Laplacian of $G$, we have
    \[
        \bra{\psi} \cal{L} \ket{\psi} \leq 2\eps~.
    \]
    On the other hand, the spectral gap $\Delta$ of $\cal{L}$ satisfies
    \[
        \| \ket{\psi} - \ket{\overline{\psi}} \|^2 \leq \frac{1}{\Delta} \bra{\psi} \cal{L} \ket{\psi} \leq \frac{2\eps}{\Delta}
    \]
    where $\ket{\overline{\psi}}$ is the vector where each entry is $\overline{\alpha} = \frac{1}{N} \sum_x \alpha_x$.

Now we analyze the success probability of $\ket{\psi}$ in the Balancedness test. This is equal to
\[
    1 - \Big \| \bra{+}^{\ot n} \cdot \ket{\psi} \Big \|^2 = 1 - N \overline{\alpha}^2 = \frac{1}{2} \| \ket{\psi} - \ket{\overline{\psi}} \|^2 \leq \eps/\Delta~.
\]
Therefore the overall probability of success is at most
\[
    \frac{1}{2} (1 - \eps) + \frac{1}{2} \max \Big \{ \eps/\Delta , 1 \Big \} \leq \frac{1}{2} (1 - \eps) + \frac{1}{2} \max \Big \{ \eps/\Delta, 1 \Big \}~.
\]
The largest this expression can get is $1 - \Delta$, corresponding to when $\eps = \Delta$.

\end{proof}

\section{$\QCMA$ lower bound}
\label{sec:qcma-lower-bound}

We now prove the $\QCMA$ lower bound on the $\Components$ problem, conditional on \Cref{conj:pseudorandomness}. In order to connect the $\Components$ problem with the pseudorandomness conjecture we first explain how Yes instances of the $\Components$ can be viewed as a simple function of a small number of independently-chosen permutations. 

\subsection{From graph oracles to raw permutations}

Recall that a Yes instance $F \in \Lyes$ describes a graph that is partitioned into two components, and each component is a disjoint union of $R$ perfect matchings. We can sample a uniformly random Yes instance as follows.

Define $\rv{X}$ to be a uniformly random permutation on $N$ elements. For every $r \in [R]$, define $\rv{Y}_r,\rv{Z}_r$ to be independently chosen permutations on $N/2$ elements. For an index $x \in [N]$ we write $\rv{X}(x),\rv{Y}_r(x),\rv{Z}_r(x)$ to denote the element that $\rv{X},\rv{Y}_r,\rv{Z}_r$ maps $x$ to. Let $\rv{P}$ denote the product random variable $\rv{X}\rv{Y}_1 \rv{Z}_1 \cdots \rv{Y}_R \rv{Z}_R$. %

Now for every instantiation $P$ of the random variable $\rv{P}$, define the graph oracle $F_P$ as follows:
\begin{longfbox}[breakable=false, padding=1em, margin-top=1em, margin-bottom=1em]
\begin{algorithm}\label{algo-plant}{ Graph oracle $F_P$ corresponding to raw permutation $P$}
\end{algorithm}

\begin{enumerate}
    \item If $X^{-1}(x) \leq N/2$, set $u = X^{-1}(x)$ and $W = Y_r$. Otherwise set $u = X^{-1}(x) - N/2$ and $W = Z_r$. 
    \item Set $s = W^{-1}(u)$. 
    \item Set $v = \begin{cases} W(s+1) & \text{if $s$ odd} \\
        W(s-1) & \text{if $s$ even} \end{cases}$. 
    \item If $X^{-1}(x) \leq N/2$, then output $X(v)$. Otherwise output $X(v + N/2)$. 
\end{enumerate}
\end{longfbox}

Let $\rv{F_P}$ denote the random variable corresponding to $F_P$; note that $\rv{F_P}$ is a function of $\rv{P}$.

\begin{claim}
\label{clm:raw-permutations-distribution}
    For all $P = X Y_1 Z_1 \cdots Y_R Z_R$, the corresponding oracle $F_P$ is a graph oracle in $\Lyes$. Furthermore, the distribution of $\rv{F_P}$ is uniform over $\Lyes$.
\end{claim}
\begin{proof}
We can view the permutations in $P$ as specifying a Yes instance as follows: $X$ specifies a partition $(S,\overline{S})$ by letting $S$ be the image of $\{1,2,\ldots,N/2\}$ under $X$ and letting $\overline{S}$ be the complement. For each $r \in [R]$, the perfect matching on $S$ (resp. $\overline{S}$) comes from mapping back from $S$ (resp. $\overline{S}$) to $[N/2]$ via $X^{-1}$, generating the perfect matching on $[N/2]$ via $Y_r$ (resp.  $Z_r$) by mapping $Y_r(i)$ (resp.  $Z_r(i)$) to $Y_{r}(i+1)$ (resp.  $Z_r(i+1)$) for all odd $i \in [N/2]$, and  by mapping $Y_r(i)$ (resp.  $Z_r(i)$) to $Y_{r}(i-1)$ (resp.  $Z_r(i-1)$) for all even $i \in [N/2]$, and then mapping back to $S$ (resp. $\overline{S}$) via $X$.

Sampling the random variable $\rv{F_P}$ is the same as the following: sample a uniformly random subset $S \subseteq [N]$ of size exactly $N/2$, for every $r$ sample a random perfect matching $\pi_{r,S}$ on $S$ and $\pi_{r,\overline{S}}$ on $\overline{S}$, and let $F(r,\cdot) = \pi_{r,S} \pi_{r,\overline{S}}$. We will prove that this  generates a uniformly random element of $\Lyes$ by showing that the number $|\{P:F=F_P\}|$ is the same for every $F \in \Lyes$. 
Consider a fixed $F  \in \Lyes$.  For all $F \in \Lyes$,  there are $(\frac{N}{2})!\cdot (\frac{N}{2})!$ possible choices for the permutation $X$. For each choice of $X$, for all $ r \in [R]$,  there are exactly $\frac{N}{2}\cdot (\frac{N}{2}-2) \cdots 2= 2^{N/4}\cdot (\frac{N}{4})!$ possible permutations $Y_r$ (resp. $Z_r$) that would induce the matching $\pi_{r,S}$ (resp. $\pi_{r,\bar{S}}$). Thus for every $F  \in \Lyes$, the number of $P$ such that $F_P=F$ is  $(\frac{N}{2})!\cdot (\frac{N}{2})\cdot (2^{N/4}\cdot (\frac{N}{4})!)^{2R}$.
 
Since the number $|\{P:F=F_P\}|$ is the same for every $F \in \Lyes$, and the distribution of $\rv{F_P}$ is uniform over its support, this implies that the distribution of $\rv{F_P}$ is uniform over $\Lyes$. 

\end{proof}

We call $P = XY_1 Z_1 \cdots Y_R Z_R$ the \emph{raw permutations} that induce the \emph{$\Components$ oracle} $F_P \in \Lyes$. Every algorithm that queries a $\Components$ oracle can be viewed as an algorithm that queries some underlying raw permutations. 

\begin{proposition}
    \label{prop:from-Lyes-to-raw}
    Let $A$ be a $T$-query algorithm that queries a $\Components$ oracle. Then there exists a $8T$-query algorithm $\hat{A}$ such that for every raw permutations $P$, 
    \[
        \Pr[A^{F_P} = 1] = \Pr[\hat{A}^P = 1]~.
    \]
\end{proposition}
\begin{proof}
    The algorithm $\hat{A}$ simulates $A$ as follows: whenever $A$ makes a query to the graph oracle $F_P$, the algorithm $\hat{A}$ queries the permutation oracles $X, Y_r, Z_r$ and their inverses according to the definition of $F_P$. Note that by construction each call to $F_P$ makes $8$ queries to the permutation oracles $X, Y_r, Z_r$ and their inverses. (The reason it is $8$ is due to implementing different oracle queries conditional on whether previous queries were larger/smaller than $N/2$, even/odd, etc.)

\end{proof}

\subsection{The lower bound}

\begin{theorem}
\label{thm:qcma_lower_bound}
   If \Cref{conj:pseudorandomness} is true, then there is no $\QCMA$ verifier $V$ for the $\Components$ problem that makes $T = N^{o(1)}$ queries and uses a classical witness of length $Q = N^{o(1)}$.
\end{theorem}

Suppose for sake of contradiction that there exists a $\QCMA$ verifier $V$ for deciding the $\Components$ problem making $ T = N^{o(1)}$  queries to a graph oracle $F$, receives a witness $w$ of length $Q = N^{o(1)}$, and decides the $\Components$ problem. By sequential amplification, we can assume without loss of generality that for $ F \in \Lyes$, there exists a classical witness $w$ such that $V^F(w)$ accepts with probability at least $99/100$, and for all $F \in \Lno$, for all witnesses $w$, the verifier $V^F(w)$ accepts with probability at most $1/100$.

Combining \Cref{clm:raw-permutations-distribution} and \Cref{prop:from-Lyes-to-raw} we have that, for all $F \in \Lyes$ there exist raw permutations $P$ such that for all strings $w \in \{0,1\}^Q$, 
    \begin{equation}
        \label{eq:qcma-raw-to-components}
        \Pr[V^{F}(w) = 1] = \Pr[\hat{V}^P(w) = 1]
    \end{equation}
    and $\hat{V}$ makes $8T$ queries to $P$. For all witnesses $w \in \{0,1\}^Q$, define
\[
\cal{L}_w := \Big \{ F \in \Lyes : \Pr[V^F(w) = 1] \geq \frac{99}{100} \Big \}~,
\]
i.e., the set of oracles that are accepted with probability at least $99/100$ using witness $w$. 
Let $w^*$ denote a witness 
such that $| \cal{L}_{w^*} |$ is maximized. Let $V_*$ and $\hat{V}_*$ denote the algorithms $V$ and $\hat{V}$, respectively, with $w^*$ hardcoded as the witness. Note that $V_*$ and $\hat{V}_*$ can be viewed as ``$\mathsf{BQP}$ query algorithms'', since they only make queries to an oracle and does not take any witness as input. 

Let $\cal{P}_{w^*} = \{ P : F_P \in \cal{L}_{w^*} \}$. These are the set of ``raw permutations'' whose corresponding $F$'s are in $\cal{L}_{w^*}$. By \Cref{clm:raw-permutations-distribution}, letting $\rv{P}$ denote the uniform distribution over raw permutations consistent with some $\Components$ oracle, the random variable $\rv{F}_{\rv{P}}$ is uniform over $\Lyes$. Therefore
\[
    \frac{|\cal{P}_{w^*}|}{|\cal{P}|} = \frac{|\cal{L}_{w^*}|}{|\Lyes|} \geq 2^{-Q}
\]
where $\cal{P}$ denotes the support of $\rv{P}$. Let $\rv{P}^*$ denote the random variable that is $\rv{P}$ conditioned on $\cal{P}_{w^*}$. 

\paragraph{Decomposing $\rv{P}^*$ into dense distributions.} The next lemma shows that the query algorithm $\hat{V}_*$ cannot distinguish between raw permutations sampled from $\rv{P}^*$ and raw permutations sampled from a (convex combination of) \emph{dense} distributions. This only depends on the fact that the random variable $\rv{P}^*$ has high min-entropy, and does not depend on the query complexity of $\hat{V}_*$.

\begin{lemma}
\label{lem:qcma-to-dense}
    For all $\eps,\delta > 0$ there exists a mixture of  random variables $\{(p_\ell,\rv{D}^{(\ell)})\}_\ell$ 
 where each $\rv{D}^{(\ell)}$ and $\rv{D}^{(\ell)^{-1}}$ is $(k,\delta)$-dense with respect to $\rv{P}$ and $k = O(Q + \log 1/\epsilon)/\delta$, such that
    \[
        \Big | \E_{P \sim \rv{P}^*} \Pr[\hat{V}_*^P = 1] - \sum_\ell p_\ell \E_{D \sim \rv{D}^{(\ell)}} \Pr[\hat{V}_*^D = 1] \Big| \leq \eps~.
    \]
   
\end{lemma}
\begin{proof}
Note that
\[
    \Hmin(\rv{P}^*) \geq \log |\cal{P}| - |w^*| = \log |\cal{P}| - Q~.
\]
We invoke the lemmas about decomposition of high min-entropy sources from \Cref{sec:min-entropy-decompositions}. Since $\rv{P}$ is the concatenation of $2R+1$ uniformly random and independent permutations (where the first permutation is on $[N]$ and the remainder are on $[N/2]$), it is nice. 

    By \Cref{lem:high-min-entropy-decomposition} and \Cref{clm:random-permutation-entropy-vs-set}, we have that $\rv{P}^*$ is $\eps$-close in total variation distance to a convex combination of random (tuples of) permutations $\{ \rv{D}_\ell \}_\ell$ such that both $\rv{D}_\ell$ and the inverse $\rv{D}_\ell^{-1}$ are $(k,\delta)$-dense with respect to $\rv{P}$.
    Letting $\{ p_\ell \}_\ell$ denote the weights of the convex combination, and using the fact that the total variation distance between two probability distributions upper bounds the difference between the probabilities of every event, the lemma follows. 
\end{proof}

We set $\eps = \frac{1}{100}$ and set $\delta = 1/\poly(T,\log N)$ for a sufficiently large polynomial to be determined later. By \Cref{lem:qcma-to-dense} and averaging, there exists an $\ell$ such that
\[
    \E_{D \sim \rv{D}^{(\ell)}} \Pr[\hat{V}_*^D = 1] \geq  \E_{P \sim \rv{P}^*} \Pr[\hat{V}_*^P = 1] - \eps \geq \frac{98}{100}~.
\]

From now on we fix this $\ell$ and simply write $\rv{D}$ and $\rv{D}^{-1}$. Since $\rv{D}$ and $\rv{D}^{-1}$ are both $(k,\delta)$-dense, there exist $k$ fixed coordinates of $\rv{D}$ (and $\rv{D}^{-1}$) that are fixed to a deterministic value and the rest of the coordinates are dense. Since $\rv{D}$ is a sequence of raw permutations $XY_1 \cdots Y_R Z_1 \cdots Z_R$ and the permutation are thought of as strings ($X$ is a string in $[N]^N$ and each of the $Y_r,Z_r$ are strings in $[N/2]^{N/2}$), the fixed coordinates of $\rv{D}$ correspond to the following:
\begin{enumerate}
    \item A set $\Gamma_X \subseteq [N]$ such that the values of $X(\Gamma_X)$ are fixed
    \item Sets $\Gamma_{Y_r} \subseteq [N/2]$ for every $r$ 
    such that the values of $Y_r(\Gamma_{Y_r})$ are fixed, and similarly for $\Gamma_{Z_r}$.
\end{enumerate}
By definition $|\Gamma_X| + \sum_{r=1}^R |\Gamma_{Y_r}| + |\Gamma_{Z_r}| \leq k$. Let $\Gamma$ denote the data that includes the fixed coordinates $\Gamma_X,\Gamma_{Y_r},\Gamma_{Z_r}$ of $\rv{D}$ as well as the values they are fixed to.

\paragraph{Invoking the pseudorandomness conjecture.} We now invoke the quantum pseudorandomness conjecture (\Cref{conj:pseudorandomness}). Let $\rv{P}_{\Gamma}$ be the random variable corresponding to $\rv{P}$ but conditioned on the same $k$ coordinates as $\rv{D}$. Since $\rv{D}, \rv{D}^{-1}$ are both $(k,\delta)$-dense, \Cref{conj:pseudorandomness} and  \Cref{lem:conjhelper} imply that
\begin{align*}
\E_{P \sim \rv{P}_{\Gamma}} \Pr[\hat{V}_*^p = 1] &\geq  \E_{D \sim \rv{D}} \Pr[\hat{V}_*^D = 1] - \poly(T,\log N)\poly(\delta)  \\
&\geq \frac{98}{100} - \poly(T,\log N)\poly(\delta)~,     
\end{align*}
We assume the following properties of the fixings $\Gamma_X,\Gamma_{Y_1},\ldots,\Gamma_{Y_R},\Gamma_{Z_1},\ldots,\Gamma_{Z_R}$:
\begin{enumerate}
    \item All the fixings of the $Y_r$ and $Z_r$ occur in pairs: for example if $i \in \Gamma_{Y_r}$, then also $i + 1 \in \Gamma_{Y_r}$ if $i$ is odd and $i - 1 \in \Gamma_{Y_r}$ if $i$ is even. 
    \item $\Gamma_X \cap \{1,2,\ldots,N/2\} = \Gamma_{Y_1} \cup \Gamma_{Y_2} \cup \cdots \cup \Gamma_{Y_R}$ and similarly $\Gamma_X \cap \{N/2 + 1,\ldots,N\} = (\Gamma_{Z_1} \cup \cdots \cup \Gamma_{Z_R}) + N/2$ (where we mean adding $N/2$ to every element of the set). 
\end{enumerate}
We can make this assumption without loss of generality, by adding at most $O(Rk)$ coordinates to each of the sets $\Gamma_X,\Gamma_{Y_r},\Gamma_{Z_r}$, and decomposing $\rv{P}_\Gamma$ further into a convex combination where those coordinates are fixed.

\paragraph{From raw permutations back to graph oracles.} By the correspondence between raw permutations and graph oracles specified in \Cref{clm:raw-permutations-distribution}, the fixed data $\Gamma$ (satisfying the assumptions listed above) gives rise to the following object:
\begin{definition}[Graph fixing data]
We call a triple $\rho = (A,B,H)$ the \emph{graph fixing data corresponding to $\Gamma$} if and only if
\begin{enumerate}
    \item $A = X(\Gamma_X \cap \{1,2,\ldots,N/2\}), B = X(\Gamma_X \cap \{N/2 + 1,\ldots,N\})$
    \item For all $r \in [R]$, for all pairs $\{ i, i+1 \} \in \Gamma_{Y_r}$ where $i$ is odd, then $(u,v,r) \in H$ where $u = X(Y_r(i))$ and $v = X(Y_r(i+1))$.
    \item For all $r \in [R]$, for all pairs $\{ i, i+1 \} \in \Gamma_{Z_r}$ where $i$ is odd, then $(u,v,r) \in H$ where $u = X(Z_r(i)+N/2)$ and $v = X(Z_r(i+1)+N/2)$.
\end{enumerate}
\end{definition}
\noindent In other words, the graph fixing data $\rho = (A,B,H)$ specifies two subsets of vertices $A,B$ that should be in different components, and $H$ is a set of (colored) edges that are either in $A \times A$ or $B \times B$. 

We let the set of all graph oracles in $\Lyes$ that are consistent with the graph fixing data $\rho$ be denoted by
\[
    \Lyes^\rho := \Set{ F \in \Lyes : \begin{array}{l}  \text{$A$ and $B$ are in different components in $F$}  \\ F(r,u) = v, \, F(r,v) = u \, \text{ for all } (u,v,r) \in H \end{array}}~.
\]
Define the random variable $\rv{G}_\rho = \rv{F}_{\rv{P}_\Gamma}$. This is a random graph oracle, and since $\Gamma$ induces the graph fixing data $\rho$ we have that $\rv{G}_\rho$ takes values in $\Lyes^\rho$. Since $\rv{P}_\Gamma$ is uniform over all  that satisfy the fixed data $\Gamma$, then the induced random variable $\rv{G}_\rho$ is uniform over all graph oracles satisfying the fixed data $\rho$, and therefore $\rv{G}_\rho$ is uniform over $\Lyes^\rho$. Therefore from \Cref{prop:from-Lyes-to-raw} we have
\[
\E_{P \sim \rv{P}_\Gamma} \Pr[\hat{V}_*^P = 1] = \E_{P \sim \rv{P}_\rho} \Pr[V_*^{F_P} = 1] = \E_{G \sim \Lyes^\rho} \Pr[V_*^G = 1]~.
\]
By setting $\delta$ to be $1/\poly(T,\log N)$ for some appropriately large polynomial $\poly(\cdot)$ we get that the number of fixed coordinates $k  = \poly(Q,T,\log N)$ and putting everything together we have that
\[
 \E_{G \sim \Lyes^{\rho}} \Pr[V_*^G = 1] \geq \frac{98}{100}-\poly(T,\log N) \cdot \poly(\delta) \geq \frac{97}{100}~.
\]

On the other hand, the algorithm $V_*$ accepts every instance $G \in \Lno$ with probability at most $\frac{1}{100}$. In particular for all $G \in \Lno^\rho$ we have that $\Pr[V_*^G = 1] \leq \frac{1}{100}$ where we analogously define
\[
    \Lno^\rho := \Set{ F \in \Lno : \begin{array}{l}  F(r,u) = v, \, F(r,v) = u \, \text{ for all } (u,v,r) \in H \end{array}}~.
\]
Note that the difference between $\Lyes^\rho$ and $\Lno^\rho$ is that $\Lno^\rho$ does not insist that the vertices $A,B$ are in different components (because there is only a single connected component). 
We thus have a $T$-query algorithm $V_*$ that distinguishes between the uniform distribution over $\Lyes^\rho$ and \emph{all} $\Lno^\rho$ with high probability.

\paragraph{Reducing to the Ambainis, Childs, and Liu lower bound.}  We now show that this contradicts the expansion testing lower bound of Ambainis, Childs, and Liu~\cite{ambainis2011quantum}. This is done via reduction: any algorithm for distinguishing between (the uniform distributions over) $\Lyes^\rho$ and $\Lno^\rho$ can be converted into an efficient algorithm that distinguishes between (the uniform distributions over) $\Lyes$ and $\Lno$, which in turn contradicts \Cref{lem:modified_acl}. We note that this proof strategy is heavily inspired by the $\QCMA$ lower bound proof of Natarajan and Nirkhe~\cite{natarajan2023distribution}. 

\begin{lemma}
\label{lem:bqpred}
Let $\eps > 0$, let $\rho$ denote graph fixing data, and let $V_*$ denote a $T$-query algorithm such that 
 \[
 \E_{G \sim \Lyes^{\rho}}\Pr[V_*^G=1] \geq \frac{3}{4} + 2\eps~,
\]
and furthermore for every $G \in \Lno^\rho$, 
\[
\Pr[V_*^G = 1] \leq \eps~.
\]
Then, for $L = O(\log \Delta^{-1} \log (N(|A \cup B|))$ there exists a $(T + (|A \cup B| \cdot L))$-query algorithm $M$ such that 
 \[
 \Big|\E_{G \sim \Lyes}\Pr[M^G=1]-\E_{G \sim \Lno}\Pr[M^G=1]\Big|\geq \eps - \frac{2}{100} - \frac{9|A \cup B|^2}{N}
 \]
\end{lemma}

Assuming \Cref{lem:bqpred} for now, we complete the proof of \Cref{thm:qcma_lower_bound}. Using that $|A \cup B| = O(Q \, \poly(T,\log N))$, \Cref{lem:bqpred} implies the existence of an algorithm $M$ that makes $T'$ queries for
\[
    T' = T + O(Q \, \poly(T,\log N))
\]
and distinguishes between the uniform distributions over $\Lyes$ and $\Lno$ with bias at least
\[
    \frac{1}{100} - O \Big( \frac{Q^2 \, \poly(T,\log N)}{N}\Big)~.
\]
If both $Q,T = N^{o(1)}$ then the bias is at least $\frac{1}{200} > \frac{1}{1000}$ for sufficiently large $N$, and furthermore the number of queries $T'$ is $N^{o(1)}$. This contradicts the lower bound from \Cref{lem:modified_acl}. Thus there cannot be a $\QCMA$ verifier for the $\Components$ problem with query complexity $T$ and witness length $Q$. \\

\medskip

We now prove \Cref{lem:bqpred}. The main intuition for this reduction is that knowing the graph fixing data $\rho$ is not helpful in distinguishing between the Yes and No cases. Given an algorithm $V_*$ for distinguishing between $\Lyes^\rho$ and $\Lno^\rho$, we design an algorithm $M$ that, given a graph oracle $G$, ``plants'' the edges specified by $\rho$ into $G$ to obtain either an instance $G'$ of $\Lyes^\rho$ (if $G\in \Lyes$) or $\Lno^\rho$ (if $G \in \Lno$), and then runs $V_*$ with the oracle $G'$. The planting procedure is both query efficient and also maps the uniform distribution over $\Lyes$ to $\Lyes^\rho$.

We present the algorithm in \Cref{fig:m}. We use the following notation. Let $\cal{C} \subseteq [N] \times [N] \times [R]$, and let $\pi:[N] \to [N]$ be a permutation. Define the function
\[
    G_{\cal{C},\pi}(r,u) = \begin{cases} 
        v & \text{ if $(u,v,r) \in \cal{C}$ or $(v,u,r) \in \cal{C}$} \\
        \pi^{-1}(G(r,\pi(u))) & \text{ otherwise}
    \end{cases}~.
\]
By construction $G_{\cal{C},\pi}$ is a graph oracle.  Suppose that $G_{\cal{C},\pi}(r,u) = v$. In one case, this means that $(u,v,r) \in \cal{C}$ or $(v,u,r) \in \cal{C}$, in which case $G_{\cal{C},\pi}(r,v) = u$. In the other case,  we have $\pi(v) = G(r,\pi(u))$, in which case $\pi(u) = G(r,\pi(v))$ because $G$ is a graph oracle.

\begin{figure}
\begin{longfbox}[breakable=false, padding=1em, margin-top=1em, margin-bottom=1em]
\begin{algorithm}\label{algo-plant} Algorithm $M$ to distinguish uniform distribution over $\Lyes$ from $\Lno$.
\end{algorithm}

Let $\rho = (A,B,H)$, and set $L = O \Big(\frac{\log N (|A| + |B|) }{\log 1/\Delta} \Big)$ such that $N (|A| + |B|) \Delta^L \leq  \frac{1}{100}$. \\

\texttt{/* Perform random walks to select vertices */} 
\begin{enumerate}
    
    \item Sample a uniformly random vertex $v_A \in [N]$ and perform a random walk $W$ of length $|A| \cdot L$ starting at $v_A$ according to the graph oracle. Let $E = (e_1,\ldots,e_{|A|})$ where $e_a$ denotes the $aL$'th vertex on the walk $W$. 
    \item Sample a uniformly random vertex $v_B \in [N]$ and perform another random walk $W$ of length $|B| \cdot L$ starting at $v_B$ according to the graph oracle. Let $F  = (f_1,\ldots,f_{|B|})$ where $f_b$ denotes the $bL$'th vertex on the walk $W$.
\end{enumerate}

\texttt{/* Sample a random relabeling of vertices that maps $(A,B)$ to $(E,F)$ */}
\begin{enumerate}
\setcounter{enumi}{2}
    \item Sample a uniformly random permutation $\pi$ on $[N]$ conditioned on $\pi(a_i) = e_i$ and $\pi(b_i) = f_i$ where $(a_i)_i$ and $(b_i)_i$ is some ordering of $A,B$.
\end{enumerate}

\texttt{/* Plant edges on $E,F$ according to $\rho$ */}
\begin{enumerate}
\setcounter{enumi}{3}
    \item $\cal{C} \leftarrow \mathrm{Reconnect}_\rho(\pi,E,F)$.
\end{enumerate}

\texttt{/* Simulate the verifier $V_*$ with the planted oracle */} 
\begin{enumerate}
\setcounter{enumi}{4}
    \item Simulate $V_*$. Whenever $V_*$ makes a query $(r,u)$ to the graph oracle $G$, return $G_{\cal{C},\pi}(r,u)$ to $V_*$.  When $V_*$ terminates, return its output.
    
\end{enumerate}

\vspace{10pt}

\begin{algorithm}\label{algo-reconnect} $\mathrm{Reconnect}_\rho(\pi,E,F)$.
\end{algorithm}
\begin{enumerate}
    \item $\cal{C} \leftarrow \emptyset$.
    \item For all $(u,v,r) \in H$: \qquad \qquad \texttt{/* Plant edge and reconnect */} 
    \begin{enumerate}
        \setcounter{enumii}{1}
        \item Query $x \leftarrow G_{\cal{C},\pi}(r,u)$ and $y \leftarrow G_{\cal{C},\pi}(r,v)$.
        \item $\cal{C} \leftarrow \cal{C} \cup \{ (u,v,r), (x,y,r) \}$.
    \end{enumerate}
    \item Output $\cal{C}$.
\end{enumerate}
\end{longfbox}
\caption{The algorithm $M$ in the proof of \Cref{lem:bqpred}.}
\label{fig:m}
\end{figure}

The random walks in Steps 1 and 2 are performed as follows. Starting a vertex $v_0$, the algorithm $M$ samples a random color $r \in [R]$ and queries the graph oracle $G$ for the next vertex $v_1 = G(r,v_0)$. Then with probability $1/2$ it stays at $v_0$, and otherwise it moves to $v_1$. (This is also known as a \emph{lazy random walk} in the graph theory literature). It then continues for the requisite number of steps. Finally, the algorithm obtains a sequence of vertices $W = (v_0,v_1,\ldots)$, and takes the vertices that are spaced every $L$ steps to get the set $E$. The algorithm performs another random walk to get the set $F$. The algorithm hopes that if the graph oracle $G$ is a Yes instance, then the second random walk starts in a different component than the first walk. This occurs with probability $1/2$, which is sufficient for our reduction to go through.

The algorithm chooses a random relabeling $\pi$ that maps the fixed vertices specified in $A$ (which is part of the graph fixing data $\rho$) to $E$ and the fixed vertices in $B$ to $F$. 

Next, the algorithm will ``plant'' the fixed edges according to $\rho$ using the $\mathrm{Reconnect}$ subroutine. For each edge-to-be-fixed $(u,v,r) \in H$, the algorithm queries the modified oracle $G_{\cal{C},\pi}$ (which has the previous edges planted already) to get the vertices $x,y$ that are matched to $u,v$, respectively, with color $r$. Then algorithm then adds the tuples $(u,v,r)$ and $(x,y,r)$ to the ``change list'' $\cal{C}$, which records that fact that $u$ is supposed to be connected to $v$, and $x$ is supposed to be connected to $y$. 

After constructing the change list $\cal{C}$, the algorithm $M$ simulates the algorithm $V_*$. Whenever $V_*$ queries $G(r,x)$, the algorithm $M$ queries the modified oracle $G_{\cal{C},\pi}(r,x)$.  

\paragraph{Analyzing the algorithm $M$.} The next claim establishes that $G_{\cal{C},\pi}$ belongs to either $\Lyes^\rho$ or $\Lno^\rho$, provided that the randomly sampled vertices $E,F$ satisfy some conditions.

\begin{claim}
\label{clm:new-G}
    Suppose $G \in \Lyes$, and let $E,F$ sampled in Steps 1 and 2 of the algorithm $M$ be such that they consist of distinct vertices and belong to separate components. Then for all $\pi$ and the change list $\cal{C} = \mathrm{Reconnect}_\rho(\pi,E,F)$, the graph oracle $G_{\cal{C},\pi}$ belongs to $ \Lyes^\rho$. On the other hand, if $G \in \Lno$, then $G_{\cal{C},\pi} \in \Lno^\rho$ for all $\pi$ and all change lists $\cal{C}$.
\end{claim}
\begin{proof}
    Suppose that $E,F$ consist of distinct vertices. Then regardless of whether $G \in \Lyes$ or $\Lno$, we have the following: for all relabelings $\pi$, for all $(u,v,r)$, we have that $(\pi(u),\pi(v),r) \in \cal{C}$ by definition of the $\mathrm{Reconnect}$ subroutine. Thus $G_{\cal{C},\pi}(r,u) = \pi^{-1}(\pi(v)) = v$ and similarly $G_{\cal{C},\pi}(r,v) = u$. 

    When $G \in \Lyes$ and $E,F$ belong to different components, this means that the vertices $A = \pi^{-1}(E)$ and $B = \pi^{-1}(F)$ belong to different components. Thus $G_{\cal{C},\pi} \in \Lyes^\rho$. When $G \in \Lno$, we have that $G_{\cal{C},\pi} \in \Lno^\rho$.
\end{proof}

Now we argue that the \emph{distribution} of $G_{\cal{C},\pi}$ induced by a uniformly random $G \sim \Lyes$ (conditioned on $E,F$ satisfying the conditions of \Cref{clm:new-G}) is close to uniform over $\Lyes^\rho$, respectively. Consider a modified algorithm $\widetilde{M}_{E,F}$ that is identical to $M$ except it doesn't perform Steps 1 and 2 of $M$ (i.e., the random walks), and instead takes the vertex sequences $E,F$ as input. 
\begin{lemma}[Random walks are close to uniform]
\label{lem:closeness}
The following bounds hold.
\begin{itemize}
    \item Let $G \in \Lyes$ have partition $S \cup \overline{S}$ into two connected components, and suppose that the subgraphs induced on the components are each $\Delta$-expanding. Then 
    \[
        \Big | \E_{E \sim S,F \sim \overline{S}} \Pr [ \widetilde{M}^G_{E,F} = 1] - \Pr[M^G = 1] \Big| \leq \frac{3}{4} + \frac{1}{100} + \frac{2|A \cup B|^2}{N}
    \]
    where $E$ is a sequence of $|A|$ vertices sampled without replacement at random from $S$, and $F$ is a sequence of $|B|$ vertices sampled without replacement at random from $\overline{S}$. 
    \item Let $G \in \Lno$ be such that the corresponding graph be a $\Delta$-expander. Then 
    \[
        \Big | \E_{E, F \sim [N]} \Pr [ \widetilde{M}^G_{E,F} = 1] - \Pr[M^G = 1] \Big| \leq \frac{1}{100} + \frac{2|A \cup B|^2}{N}
    \]
    where $E$ is a sequence of $|A|$ vertices sampled without replacement at random from $[N]$, and $F$ is a sequence of $|B|$ vertices sampled without replacement at random from $[N]\setminus E$. 
\end{itemize}
\end{lemma}
For now we defer the proof of this lemma and proceed. By construction, for fixed graph oracle $G$, for all sequences $E,F$, 
\[
    \Pr[\widetilde{M}^G_{E,F} = 1] = \E_{\pi: \pi(A) = E, \pi(B) = F} \Pr[V_*^{G_{\cal{C},\pi}} = 1]
\]
where the probabilities are over the randomness of the algorithm $M$. 

Observe that a uniformly random $G \sim \Lyes$ has two connected components on $N/2$ vertices that are themselves unions of $R$-random matchings. Thus by \Cref{lem:random-graphs}, both components fail to be $\Delta$-expanding with probability at most $\frac{4}{N}$. Combined with the first item of \Cref{lem:closeness}, we have
\begin{align}
    \E_{G \sim \Lyes} \Pr[M^G = 1] &\geq \E_{G \sim \Lyes} \E_{\substack{E \sim S \\ F \sim \overline{S}}} \E_{\substack{\pi: \\ \pi(A) = E, \\ \pi(B) = F}} \Pr[V_*^{G_{\cal{C},\pi}} = 1] - \frac{3}{4} - \frac{1}{100} - \frac{4}{N} - \frac{2|A \cup B|^2}{N}~.
    \label{eq:bqpred-1}
\end{align}
But then 
\begin{align}
   \E_{G \sim \Lyes} \E_{\substack{E \sim S \\ F \sim \overline{S}}} \E_{\substack{\pi: \\ \pi(A) = E, \\ \pi(B) = F}} \Pr[V_*^{G_{\cal{C},\pi}} = 1] = \E_{(E,F) \sim \mu} \, \E_{G \sim \Lyes(E,F)}  \E_{\substack{\pi: \\ \pi(A) = E, \\ \pi(B) = F}} \Pr[V_*^{G_{\cal{C},\pi}} = 1]
   \label{eq:bqpred-2}
\end{align}
where $\Lyes(E,F) \subseteq \Lyes$ is the subset of Yes instances such that one component has vertices $E$ and the other component has vertices $F$. The distribution $\mu$ is chosen such that sampling $(E,F) \sim \mu$, and then a uniformly random $G \sim \Lyes(E,F)$, yields a uniformly random $G \sim \Lyes$. We now state a lemma (to be proved later) establishing that $G_{\cal{C},\pi}$ is uniformly distributed in $\Lyes^\rho$.
\begin{lemma}[Reconnect preserves the uniform distribution]
\label{lem:metagraph}
Consider the distribution over $G_{\cal{C},\pi}$ induced as follows: first sample $(E,F) \sim \mu$, then sample $G \sim \Lyes(E,F)$, then sample $\pi$ such that $\pi(A) = E,\pi(B) = F$, and then compute $\cal{C} = \mathrm{Reconnect}_\rho(\pi,E,F)$. This distribution is exactly equal to the uniform distribution over $\Lyes^\rho$. 
\end{lemma}
\Cref{lem:metagraph} combined with equations~\eqref{eq:bqpred-1} and~\eqref{eq:bqpred-2} implies that
\[
  \E_{G \sim \Lyes} \Pr[M^G = 1]  \geq \E_{G \sim \Lyes^\rho}   \Pr[V_*^{G} = 1] - \frac{3}{4} - \frac{1}{100} - \frac{6|A \cup B|^2}{N} \geq 2\eps - \frac{1}{100} - \frac{6|A \cup B|^2}{N}. 
\]
where we used the assumption about the behavior of $V_*$ on $\Lyes^\rho$ from the statement of \Cref{lem:bqpred}. 

Similarly, by the second item of \Cref{lem:closeness}, by \Cref{clm:new-G}, and the fact that $G$ is not $\Delta$-expanding with probability at most $\frac{1}{N}$, we have
\begin{align}
    \E_{G \sim \Lno} \Pr[M^G = 1] &\leq \E_{G \sim \Lno} \E_{E,F \sim [N]} \E_{\substack{\pi: \\ \pi(A) = E, \\ \pi(B) = F}} \Pr[V_*^{G_{\cal{C},\pi}} = 1] +  \frac{1}{100}+\frac{1}{N} + \frac{2|A \cup B|^2}{N}  \notag \\
    &\leq \eps + \frac{1}{100}+\frac{1}{N} + \frac{2|A \cup B|^2}{N}~.
    \label{eq:bqpred-3}
\end{align}
where we used the fact that $V_*$ rejects every instance from $\Lno^\rho$. 

Putting everything together we get
\[
    \Big | \E_{G \sim \Lyes} \Pr[M^G = 1] - \E_{G \sim \Lno} \Pr[M^G = 1]\Big |\geq \eps - \frac{2}{100} - \frac{9|A \cup B|^2}{N}~.
\]
as desired. 

We finish the proof of \Cref{lem:bqpred} by presenting the proofs of \Cref{lem:metagraph,lem:closeness}. 

\begin{proof}[Proof of \Cref{lem:closeness}]
We prove the first item. Let $G \in \Lyes$. Consider the random walk in the algorithm $M$ (i.e., Steps 1 and 2). Let $J$ be the event that $v_A \in S$ and $v_B \in \Bar{S}$, which occurs with probability $1/4$. Observe that
\[
    \Pr[M^G = 1] =  \frac{1}{4}\Pr[M^G \mid J] + \frac{3}{4} \Pr[M^G \mid \neg J]~.
\]
Therefore by triangle inequality 
\begin{align*}
    &\Big | \E_{E \sim S,F \sim \overline{S}} \Pr [ \widetilde{M}^G_{E,F} = 1] - \Pr[M^G = 1] \Big|\\
    &\leq  \frac{1}{4}\Big | \E_{E \sim S,F \sim \overline{S}} \Pr [ \widetilde{M}^G_{E,F} = 1] - \Pr[M^G = 1\mid J]\Big|+ \frac{3}{4}\Big | \E_{E \sim S,F \sim \overline{S}} \Pr [ \widetilde{M}^G_{E,F} = 1] - \Pr[M^G = 1 \mid \neg J]\Big|\\
    &\leq \Big | \E_{E \sim S,F \sim \overline{S}} \Pr [ \widetilde{M}^G_{E,F} = 1] - \Pr[M^G = 1|J]\Big|+ \frac{3}{4}
    \end{align*}

We bound the first term. Since the algorithm $M$ can be viewed as first sampling $E,F$ according to the random walk and then running $\tilde{M}_{E,F}$, it is sufficient to show that the distribution of $E,F$ generated using random walks on $S$ and $\bar{S}$ is close to the distribution generated by sampling $E$ without replacement at random from $S$, and sampling $F$ without replacement at random from $\bar{S}$.

We first consider the random walk to generate $E$. Conditioned on event $J$, the starting point of the walk is contained in $S$ and therefore all the vertices of $E$ are solely contained in $S$. 
Using \Cref{cor:tlazyrw} with $K=|A|$ and $t=L$ and the subgraph induced on $S$ (which is assumed to be $\Delta$-expanding), the distribution of $E$ generated by the random walk is at most $|A| \cdot N \cdot \Delta^L$-close to the uniform distribution over $S^{|A|}$ (i.e., every vertex in $E$ is a independent, random vertex in $S$). This uniform distribution is in turn at most $2|A|^2/N$-close to sampling $E$ from $S$ without replacement. Similarly, the distribution of $F$ (conditioned on the event $J$) is $(|B| \cdot N \cdot \Delta^L + 2|B|^2/N)$-close to sampling $F$ from $\overline{S}$ without replacement. 

Putting both bounds together we get that
\[
  \Big | \E_{E \sim S,F \sim \overline{S}} \Pr [ \widetilde{M}^G_{E,F} = 1] - \Pr[M^G = 1|J]\Big| \leq (|A| + |B|)N\Delta^L + \frac{2(|A|^2 + |B|^2)}{N} \leq \frac{1}{100} +  \frac{2|A \cup B|^2}{N} 
\]
where we used the setting of $L$ from the description of $M$ (\Cref{algo-plant}). This concludes the proof of the first item.

The proof of the second item is similar. In the case that $G \in \Lno$, there is no partition between $S$ and $\Bar{S}$; the two random walks occur in the same component and can be viewed as one long random walk of length $(|A| + |B|)L$. By \Cref{cor:tlazyrw} once more we get that the distribution of $E,F$ in the random walk is $(N(|A| + |B|)\Delta^L + (|A|+|B|)^2/N)$-close to sampling $E,F$ from $[N]$ without replacement. Thus 
\[
        \Big | \E_{E, F \sim [N]} \Pr [ \widetilde{M}^G_{E,F} = 1] - \Pr[M^G = 1] \Big| \leq 
        \frac{1}{100} + \frac{2|A \cup B|^2}{N}
    \]
    as desired.
\end{proof}

\begin{proof}[Proof of \Cref{lem:metagraph}]
    We analyze the $\mathrm{Reconnect}_\rho$ procedure iteratively. Let $\rho = (A,B,H)$ be the graph fixing data. Let $\cal{C}_i$ denote the change list after processing the $i$'th edge. Let $\rho_i$ denote the graph fixing data $\rho$ restricted to the first $i$ edges of $H$. We also define $\cal{C}_0 = \emptyset$ and $\rho_0 = (A,B,\emptyset)$. Note that we have the following chain of inclusions:
    \[
        \Lyes(A,B) = \Lyes^{\rho_0} \supset \Lyes^{\rho_1} \supset \cdots \supset \Lyes^{\rho_{|H|}} = \Lyes^\rho
    \]
    Since $G \in \Lyes(E,F)$, after relabeling using $\pi$ we have that $G_{\cal{C}_0,\pi} \in \Lyes(A,B)$. By construction, after the $i$'th step of $\mathrm{Reconnect}_\rho$, the graph oracle $G_{\cal{C}_i,\pi}$
    belongs to $\Lyes^{\rho_i}$. 

    We first make the following simple observation about each iteration of the $\mathrm{Reconnect}$ procedure.
    \begin{claim}
    \label{clm:rec-ind}
        Let $\cal{D}_{i+1} = \cal{C}_{i+1} \setminus \cal{C}_i$. Then
       $G_{\cal{C}_{i+1},\pi} = \widetilde{G}_{\cal{D}_{i+1},e}$ where $\widetilde{G} = G_{\cal{C}_i,\pi}$ and $e$ denotes the identity permutation. 
    \end{claim}
    \begin{proof}

Observe that $G_{\cal{C}_{i+1},\pi}$ can be equivalently written as  \[
    G_{\cal{C}_{i+1},\pi}(r,u) = \begin{cases} 
        v & \text{ if $(u,v,r) \in \cal{C}_{i+1}\setminus \cal{C}_i$ or $(v,u,r) \in \cal{C}_{i+1}\setminus \cal{C}_i$}\\
          v & \text{ if $(u,v,r) \in  \cal{C}_i$ or $(v,u,r) \in  \cal{C}_i$}\\
          \pi^{-1}(G(r,\pi(u))) & \text{ otherwise}
    \end{cases}~.
\]
Now observe that the last two cases are equivalent to the definition of $G_{\cal{C}_i,\pi}$ on all inputs $(r,u)$ that do not appear in $\cal{D}_{i+1}$. Letting $\widetilde{G} = G_{\cal{C}_i,\pi}$, we get that

 \[
   \tilde{ G}_{\cal{D}_{i+1},e}(r,u) = \begin{cases} 
        v & \text{ if $(u,v,r) \in \cal{D}_{i+1}$ or $(v,u,r) \in \cal{D}_{i+1}$}\\
        \tilde{G}=G_{\cal{C}_i,\pi}(r,u) & \text{ otherwise}
    \end{cases}
\]
as desired. %
    \end{proof}

    We now show that the $i$'th iteration of $\mathrm{Reconnect}_\rho$ maps the uniform distribution over $\Lyes^{\rho_i}$ to the uniform distribution over $\Lyes^{\rho_{i+1}}$. 

    \begin{claim}
    \label{claim:preimage}
        Fix $E,F,\pi$ and an index $1 \leq i \leq |H|$. If the distribution of $G_{\cal{C}_i,\pi}$ is uniform over $\Lyes^{\rho_i}$, then the distribution of $G_{\cal{C}_{i+1},\pi}$ is uniform over $\Lyes^{\rho_{i+1}}$.  
    \end{claim}
    \begin{proof}
    \Cref{clm:rec-ind} implies that, for fixed $E,F,\pi,i$, there is a map $\varphi: \Lyes^{\rho_i} \to \Lyes^{\rho_{i+1}}$ where for all $G_i \in \Lyes^{\rho_i}$, 
    \[
        \varphi(G_i) = (G_i)_{\cal{D}_{i+1},e}
    \]
    where $\cal{D}_{i+1}$ only depends on the $(i+1)$st edge $(u,v,r)$ in the graph fixing data $\rho$, and $G_i$.\footnote{ Here, $G_i$ is representing the graph $G_{\cal{C}_i,
    \pi}$, so the permutation $\pi$ is performed within the query to $G_i$ itself.} 
    
    We show that $\varphi$ is a regular function, i.e., each element of its image $\Lyes^{\rho_{i+1}}$ has the same number of preimages. Fix $G_{i+1} \in \Lyes^{\rho_{i+1}}$. If $\varphi(G_i) = G_{i+1}$, then it must be that $\cal{D}_{i+1} = \{ (u,v,r), (x,y,r) \}$ where $x = G_i(r,u)$ and $y = G_i(r,v)$. Note that $G_{i+1}$ either differs from $G_i$ in two edges (because the edges $(u,x)$ and $(v,y)$ in $G_i$ were switched to $(u,v)$ and $(x,y)$ in $G_{i+1})$) or $G_{i+1} = G_i$ (because $(u,v)$ was an edge already). Thus we simply have to count the number of possible $x$'s in $G_{i+1}$ that could have possibly been matched with $u$ in $G_i$. The only $x$'s that are forbidden are the ones that are already matched with an edge with color $r$ in $\rho_i$. This number only depends on $\rho_{i+1}$ and is independent of $G_{i+1}$. Therefore $\varphi$ is a regular function, and thus if $G_i$ is uniform over $\Lyes^{\rho_i}$, then $\varphi(G_i)$ is uniform over $\Lyes^{\rho_{i+1}}$.

    We illustrate the essence of this proof with \Cref{fig:metagraph}.

    \end{proof}
    \Cref{lem:metagraph} now directly follows: since $G_{\cal{C}_0,\pi}$ is uniform over $\Lyes(A,B)$, then $G_{\cal{C}_1,\pi}$ is uniform over $\Lyes^{\rho_1}$, and so on, until we get the final $G_{\cal{C},\pi}$ which is uniform over $\Lyes^\rho$.

\end{proof}

\begin{figure}
{\center 
\begin{tikzpicture}
    \draw[thick] (0, 0) ellipse (4 cm and 2 cm); %
     \node at (0, 0) {$\mathcal{L}_{\mathrm{yes}}(A,B)$};
     \draw[->, thick] (-3, -1) -- (-1, -4);
    \fill (-3,-0.9) circle (1pt);
      \draw[->, thick] (-2.5, -1) -- (-1, -4); 
       \fill (-2.5,-0.9) circle (1pt);
       \foreach \x in {-2, -1.75,-1.5,-1.25} {
        \fill ( \x,-0.9) circle (1pt);}
      \draw[->, thick] (-1, -1) -- (-1, -4); 
       \fill (-1,-0.9) circle (1pt);
        \draw[->, thick] (3, -1) -- (1, -4);
     \fill (3,-0.9) circle (1pt);
      \draw[->, thick] (2.5, -1) -- (1, -4); 
       \fill (1,-0.9) circle (1pt);
       \foreach \x in {2, 1.75,1.5,1.25} {
        \fill ( \x,-0.9) circle (1pt);}
      \draw[->, thick] (1, -0.9) -- (1, -4); 
     \fill (2.5,-0.9) circle (1pt);

      \draw[thick] (0, -5) ellipse (3.5cm and 1.75cm);
      \fill (-1,-4.2) circle (1pt);
       \foreach \x in { -0.5,-0.25,0,0.25,0.5} {
        \fill ( \x,-4.2) circle (1pt);}
          \fill (1,-4.2) circle (1pt);
       \node at (0, -5) {$\mathcal{L}^{\rho_1}_{\mathrm{yes}}$};
       \foreach \y in {-7, -7.20,-7.4} {
        \fill (0, \y) circle (1pt);} %
        \draw[thick] (0, -9) ellipse (3 cm and 1.5cm);
       \node at (0, -9) {$\mathcal{L}^{\rho_{|H|-1}}_{\mathrm{yes}}$};
 \draw[->, thick] (-1, -9.3) -- (-1.4, -13);
    \fill (-1, -9.2) circle (1pt);
     \draw[->, thick] (-2.5, -9.3) -- (-1.4, -13);
     \fill (-2.5, -9.2) circle (1pt);
   \draw[->, thick] (1, -9.3) -- (1.4, -13);
      \fill (1, -9.2) circle (1pt);
     \draw[->, thick] (2.5, -9.3) -- (1.4, -13);
     \fill (2.5, -9.2) circle (1pt);
       \draw[thick] (0, -13) ellipse (2 cm and 1cm);
       \node at (0, -13) {$\mathcal{L}^{\rho}_{\mathrm{yes}}$};
      \fill (-1.5,-13.1) circle (1pt);
       \foreach \x in { -0.75,-0.5,-0.25,0,0.25,0.5,0.75} {
        \fill ( \x,-13.5) circle (1pt);}
        
          \fill (1.5,-13.1) circle (1pt);
       
    \foreach \x in { -2,-1.75,-1.5,-1.25} {
        \fill ( \x,-9.2) circle (1pt);}
        \foreach \x in { 2,1.75,1.5,1.25} {
        \fill ( \x,-9.2) circle (1pt);}

\end{tikzpicture}
}
\caption{This figure illustrates the iterative $\mathrm{Reconnect}_{\rho}$ procedure. The first layer represents all Yes instances conditioned on $A \in S$ and $B \in \bar{S}$, and each circle represents a graph in the corresponding set. There is an edge from a graph $G$ in layer $i$ to a graph $G'$ in layer $i+1$ if $\varphi(G)=G'$ . Each graph in every layer has the same number of incoming edges from the previous layer. Thus the figure is illustrating that the iterative $\mathrm{Reconnect}_{\rho}$ procedure starts with a  uniform distribution over $\Lyes(A,B)$, and transforms it into uniform distributions over smaller and smaller sets, with the last set being $\Lyes^{\rho}.$}
\label{fig:metagraph}
\end{figure}

\section{Unconditional $\QMA$ vs $\QCMA$ separation for an interactive game}
\label{sec:interactive-game}

In this section, we 
give an \emph{unconditional} $\QMA$/$\QCMA$ oracle separation for an interactive oracle distinguishing game.
This result leads to a simplified proof of the $\QMA$ and $\QCMA$ separation relative to a \emph{distributional} oracle.

The techniques we use to analyze the lower bound for interactive games have appeared in the literature for dealing with non-uniformity of algorithms. These are algorithms that have an unbounded preprocessing stage to produce an advice (which corresponds to the witness in the language of QMA/QCMA) that helps a bounded second stage algorithm to solve a  challenge.
The formulation of interactive games provides an alternative and clean view to the distributional oracle separation in previous works, demonstrating their conceptual similarity.

The oracle has a similar structure to the one used in \cite{natarajan2023distribution} and our previous sections, but with a few modifications.
The proof roadmap is as follows:
we first define an interactive variant of the $\QMA$/$\QCMA$ protocol, and give an oracle separation for this variant through a technique that removes the use of witness in our analysis with a sequentially repeated game (multi-instance security). Finally, we show a reduction from the oracular interactive $\QMA$/$\QCMA$ game to 
a distributional oracle $\QMA$/$\QCMA$ game.

\subsection{Multi-instance security helper lemma}
In this section, we will leverage a useful notion called multi-instance security.

\paragraph{Multi-Instance Game} 
For any security game $G$ and any positive
integer $k$, we define the multi-instance game $G_k = G^{\otimes k}$ with respect to an oracle $\cO$, where  $G^{\otimes k}$
is given as follows:

For $i \in [k]$:
\begin{enumerate}
    \item Sample fresh randomness $r_i$

    \item Compute challenge using randomness $r_i$ and $\cO$, send challenge to adversary.

    \item Adversary is given oracle access to $\cO$ and produces a quantum state $\ket{\ans_i}$ as answer. 

    \item The challenger makes a binary-outcome projective measurement on the state $\ket{\ans_i}$ to obtain outcome $b_i$, and gives back the remaining state $\ket{\ans_i'}$ to adversary. 
\end{enumerate}
In the end, the challenger outputs $(b_1, b_2 \cdots, b_k)$.

We will deploy the following lemma from \cite{chung2020tight}:

\begin{lemma}
\label{lem:multi_instance_condition}
  Let $G$ be any single-instance security game. Let $B_i$ be the random variable for $b_i$ in the multi-instance game $G_k$ as defined above, and let $X : =(B_1, \cdots, B_{i-1})$ be the random variable on the outcome of the first $(i-1)$ rounds. For any $\delta$, if for any multi-instance adversary $A$ and for any $x \in \{0,1\}$ that $Pr[X = x] > 0$, we have that:
  $$  \Pr[B_i \vert X =x] \leq \delta$$
Then we have that $\Pr[(B_1, B_2, \cdots, B_k) = 1^k] \leq \delta^k$. 
\end{lemma}

\subsection{The Interactive Oracle Distinguishing Problem}
\label{sec:interactive_oracle_distinguish}

Fix integers $n,R \in \N$ and $\eta$ is a constant where $0 < \eta \leq \frac{1}{2}$. Let $N = 2^n$. The oracles we consider will be functions $F: [R] \times \{0,1\} \times [N] \to [N]$. 
For simplicity, we take our $\eta = \frac{1}{2}$ and $\vert R \vert = N$.

\begin{definition}
    We say that an oracle $F: [R] \times \{0,1\} \times [N] \to [N]$ is an \emph{$\eta$-noisy permutation oracle} if there exists a subset $Z \subseteq [R]$ with $|Z| \geq (1 - \eta)R$ such that for all $r \in Z$, there exists a permutation $\pi$ on $[N]$ such that $F(r,0,x) = \pi(x)$ and $F(r,1,x) = \pi^{-1}(x)$ for all $x \in [N]$. We call the $r$'s in $Z$ \emph{proper} . 
\end{definition}
\noindent Note that for non-proper $r$, the function $F(r,\cdot,\cdot)$ may behave arbitrarily. 

Define the following sets of Yes and No instances of $\eta$-noisy permutation oracles. Let $F$ be an $\eta$-noisy permutation oracle. Then 

\begin{itemize}
    \item $F \in \Lyes$ if and only if there exists a subset $S \subseteq [N]$ with $|S| = N/2$ such that for all proper $r$, the permutation $\pi$ associated with $F(r,\cdot,\cdot)$ is a product of two single-cycle permutations $\sigma_1 \sigma_2$ where $\sigma_1$ is a single-cycle permutation on $S$ and $\sigma_2$ is a single-cycle permutation on $\overline{S}$. We call $(S,\overline{S})$ the \emph{partition associated with $F$} (which is uniquely determined because the proper $r$'s form the majority of $[R]$). 

    \item $F \in \Lno$ if and only if for all proper $r$, the permutation $\pi$ associated with $F(r,\cdot,\cdot)$ is a single cycle on $[N]$, and furthermore for each $x \in [N]$ and $b \in \{0,1\}$, the distribution of $F(r,b,x)$ over $r \in R$ is $\delta$-close to uniform on $[N]$.

\end{itemize}

\begin{remark}
    This above condition implies that the associated graph has spectral gap at least $1 - \delta N$
\end{remark}

\begin{remark}
   For simplicity, we denote $F(r, \cdot) := F(r,0, \cdot)$ and $F^{-1}(r, \cdot) := F(r, 1, \cdot)$ for the rest of the section.
\end{remark}

\paragraph{Interactive Oracle-Distinguishing Game}
Consider the following game $G$:
\begin{enumerate}
    \item An oracle $F$ is sampled in the following manner:
    \begin{enumerate}
        \item Sample $S \subseteq [N]$ uniformly random, where $|S| = N/2$.
        \item Sample $Z \subseteq [R]$ uniformly at random. 
        \item For each $r \in Z$, sample a random cycle $\pi_S$ on $S$ and a random cycle $\pi_{\overline{S}}$ on $\overline{S}$, and let $F(r,\cdot) = \pi_S \pi_{\overline{S}}$. 
        \item For each $r \notin Z$, sample a random cycle $\pi$ on $[N]$, and let $F(r,\cdot) = \pi$. 
    \end{enumerate}

    \item In the $\QMA$ case, the adversary sends a quantum witness $\ket{\psi}$ that depends on $F$, and in the $\QCMA$ case, the adversary sends a classical witness $w$ that depends on $F$. 

    \item The verifier, in addition to the witness, receives a challenge $r$ sampled uniformly at random.

    \item The verifier outputs a bit $b$, and wins if it correctly guesses if $F(r,\cdot)$ is a product of two cycles or one cycle.
\end{enumerate}

\begin{claim}
\label{claim:qma_correctness_interactive}
    There exists a polynomial-size quantum proof $\ket{\psi}$, verifiable using polynomially-many queries to $F$, such that the verifier wins this game with probability $1 - o(1)$.
\end{claim}
\begin{proof}
The quantum witness is the same state $\ket{\psi} := \sqrt{\frac{1}{N}}\sum_{x \in S} \ket{x} - \sqrt{\frac{1}{N}} \sum_{x \notin S} \ket{x}$ as the one used in \Cref{sec:qma_upper_bound}.

Note that the set $S, \overline{S}$ are the subset of vertices in the two disjoint cycles corresponding to the partial oracle $F(r, \cdot)$. But the $\QMA$ verifier works in a slightly different way: since the verifier receives the challenge $r^*$, it uses this $r^*$ to perform the invariance check. The $\QMA$ verifier randomly perform the following tests:
    \begin{enumerate}
        \item \textbf{Balancedness test.} Apply $H^{\otimes n}$ to $\ket{\psi}$ and check if the result is $\ket{0^n}$. If so, reject. Otherwise, continue.

        \item \textbf{Invariance test.} 
        \begin{enumerate}
        
            \item Prepare a control qubit $\reg{C}$ in the state $\ket{+}$, %
            \item Controlled on $\reg{C}$ being in the state $\ket{1}$, %
            apply the permutation $F(r^*,\cdot)$ to the register $\reg{X}$. %
            
            \item Apply $H$ to $\reg{C}$ and check if the result is $\ket{0}$. If not, reject. 
            If yes, continue.%
        \end{enumerate}
        Accept if all the above steps are accepted.
\end{enumerate}
  The correctness follows from the analysis of the $\QMA$ verifier in \Cref{sec:qma_upper_bound}.
   The soundness follows 
    by applying the same analysis in \Cref{qma_verifier_algo}, since $r^*$ is chosen uniformly at random by the challenger. We have $\Pr_{r \gets [R]}[\text{~$\QMA$ verifier accepts~}] \leq 1-\Delta$, where $\Delta$ is the spectral gap of the normalized Laplacian of the graph.

\end{proof}

\subsection{$\QCMA$ lower bound via multi-instance game}

The high-level idea of our $\QCMA$ lower bound for the interactive game is the following: for all $\QCMA$ verifiers that take in a polynomial-size classical proof $w$ and make polynomially-many queries to $F$, can only solve this with probability bounded away from $1$ (ideally, very close to $1/2$). 

One way to show this is to utilize the multi-instance analysis to remove the use of a witness by guessing the witness used in a sequentially repeated version of the original game. Consider the multi-instance game $G^{\otimes k}$, where the verifier receives $k$ independent challenges $(r_1,\ldots,r_k)$ \emph{in sequence} and has to simultaneously guess correctly on all of them.

The following claim is straightforward when $W$ is a classical witness. \footnote{When the witness is quantum, one may not be able to use the same witness for all instances of a multi-instance game due to the destructible nature of quantum states, and the loss factor will be different. But when the witness is classical, we can consider that the prover provides some $W$ with some fixed, arbitrary polynomial size of its own choice, which can be used throughout the sequential repetition.}.

\begin{claim}
\label{claim:remove_witness_with_multi_instance}
    If there is a $T$-query $\QCMA$ verifier $V$ that succeeds in the game $G$ (when it receives the $W$-length witness $w$ depending on $F$) with probability at least $\frac{1}{2} + \delta$, then there exists a proof-less verifier $V'$ that succeeds in the game $G^{\otimes k}$ with probability at least
    \[
        2^{-W} \Big ( \frac{1}{2} + \delta \Big)^{k}
    \]
    and makes $T$ queries to $F$ per round.
\end{claim}

Suppose that one can prove that any algorithm for the game $G^{\otimes k}$ that makes $\poly(\log N)$ queries per round cannot succeed with probability more than $(\frac{1}{2} + \eta)^k$ for some negligible quantity $\eta$. Then
Therefore
\[
    2^{-W} (\frac{1}{2} + \delta)^k \leq (\frac{1}{2} + \eta)^k~.
\]
This means 
\[
    \frac{1}{2} + \delta \leq 2^{W/k} (\frac{1}{2} + \eta)~.
\]
Setting $k = W/\eps$ for some proper $ \eps $ we get $\delta$ is still negligible. 
This parameter is related to $K$ and the oracle size $N$. We will set this parameter later.

Thus all that remains is to analyze the probability of success in the game $G^{\otimes k}$.

\subsection{Winning probability in a game with no witness}

\paragraph{Single-Instance Game Lower Bound}
We first show a lower bound in the single instance game.

\begin{theorem}
    The single instance game of distinguishing whether a $F(r, \cdot)$ is a Yes case or a No case requires $\Omega(\eta\, N)$ queries, to succeed with probability $1/2+\eta$ for some $\eta \in [0,1/2]$, for any $F(r, \cdot)$.
\end{theorem}

\begin{proof} 
    We apply \Cref{thm:adversary_method} in the following way. 
    Consider the following relation $R$: For some $F(r,\cdot)$ in the Yes instance, there exists two neighboring vertices $(u_1, v_1)$ in the first cycle, and two neighboring vertices $(u_2, v_2)$ in the second cycle such that if we remove the edges between $(u_1, v_1)$ and  $(u_2, v_2)$, and replace with an edge $(u_1, u_2)$ and an edge $(v_1, v_2)$, we obtain an instance in the No case.
    Each such change in the oracle $F(r, \cdot)$ for neighbouring vertices pair $(u_1, v_1)$ and  $(u_2, v_2)$ will reflect as changes in four entries of the truth table for $F(r, \cdot)$.

We analyze the following values from \Cref{thm:adversary_method}:
    \begin{enumerate}
        \item For each $F(r,\cdot)$ in the Yes instance: there are $N/2$ choices of  $(u_1, v_1)$ in the first cycle and $N/2$ choices of  $(u_2, v_2)$ in the second cycle to modify them into one No instance. So in total, there are $\Omega(N^2)$ number of No instances in relation with it.

        \item For each  $F(r,\cdot)$ in the No instance: there are $N$ choices of $(u_1, u_2)$, but once $(u_1, u_2)$ is fixed, there is only one choice of edge $(v_1, v_2)$ where $v_1$ is $(N/2-1)$ edges away from $u_1$ and $v_2$ is $(N/2-1)$ edges away from $u_2$ (on the shorter path), because our two cycles in the Yes case must have size $N/2$ each. Thus, there are $N$ number of Yes instances in relation with each No instance.

        \item For every Yes oracle and each entry $F_r(u_i)$ in the adjacency matrix, we consider the No instances different
        from it at entry $F_r(u_1)$: for every $F_r(u_1) = v_1, F^{-1}_r( v_1) = u_1$ in the Yes case, then a No instance is different at entry $F_r(u_1), F^{-1}_r,(v_1)$, then we have $|S|/2$ choices for the vertex $u_2$ such that $F_r(u_1)=u_2$. Once $u_2$ is picked, then we have 2 choices for the vertex $v_2$ such  that $(v_1,v_2)$ is an edge. So the total choices are $N$, and $l_{x,u_1,v_1}=N$.
        
        \item For every No oracle and each entry pair $F_r(u_1), F_r^{-1}(v_1)$ in truth table, we consider the Yes instances different
        from it at entry $F_r(u_1)$ (and accordingly $F_r^{-1}(v_1)$): if $F_r(u_1) \neq v_1$ in the No instance, now we modify  $F_r(u_1) = v_1, F^{-1}(v_1) = u_1$ so that it is a Yes instance. If this is a valid transformation, then we must have that $u_1$ is $(N/2-1)$ vertices from $v_1$ in the original No instance, on the shorter path, which makes $u_1$ unique. Now we must remove the edge $(u_1,u_2)$(correspondingly, entries $F_r(u_1), F^{-1}_r(u_2)$) where $u_2$ is on the longer path from $u_1$ to $v_1$; similarly, we remove edge $(v_1, v_2)$ where $v_2$ is on the longer path from $v_1$ to $u_1$; finally, add edge $(v_1,v_2)$(correspondingly, entries $F_r(v_1), F^{-1}_r(v_2)$) to get a Yes instance. Thus there is only one choice and $l_{y,u_1,v_1}' = 1$.
    \end{enumerate}
    
    If we analyze all entries where $F_r(u_1) \neq v_1 $ in the Yes instance and $F_r(u_1) = v_1$ in the No instance we get a smaller $l_{x,u_1,u_2} \cdot l_{y,u_1,u_2}'$. 
    Thus we take $l_{max} = 1 \cdot N$.
    
    Therefore we have query lower bound $\Omega(\sqrt{\frac{m \cdot m'}{l_{max}}}) = \Omega(\sqrt{\frac{N^3}{N}}) =  \Omega(N)$ for constant success probability $1/2+\eta$.

    If we like to make $\eta$ negligible (for example, $1/\sqrt{N}$), we have a query lower bound $\Omega(\eta \cdot N)$ which is still exponential. 
\end{proof}

\paragraph{Sequential repetition/Multi-instance security}

We show a sequential repetition of the interactive game using 
\Cref{lem:multi_instance_condition}.

\begin{claim}
 Let $X := (B_1, \cdots, B_{i-1})$  be the random variable for the output $b_1, \cdots, b_{i-1}$ of the first $(i-1)$ rounds of the multi-instance game, we have for any value $x$, for any $i \in [k]$: $\Pr[B_i = 1 \vert X = x] \leq 
 1/2 + 1/\sqrt{N} + O(K/N)$.  
\end{claim}

\begin{proof}
    Since each $r_i$ is independently chosen at random, and the cycle(s) each oracle $F(r_i, \cdot)$ are independently prepared. Suppose $A$ makes $q$ arbitrary queries in each challenge $i \in [k]$ (not necessarily to oracle $F(r_i,\cdot)$). In the $i$-th round, with probability $(1 - O(\frac{i}{N}))$, $r_i \neq r_j$ for all $j < i$ and the oracle $F(r_i, \cdot)$ is independent of all previous oracles $F(r_j, \cdot)$. More concretely, the suppressed small terms in $O(i/N)$ comes from the possibility that the oracles for different permutations happened to be the same and since they are much smaller than $i/N$, we can replace $O(i/N)$ with $c\cdot i/N$ for some constant $c$ for simplicity.
   Therefore, the remaining quantum state from previous rounds is independent of the oracle $F(r_i, \cdot)$, an adversary making polynomially many queries's advantage is upper bounded by $1/2+1/\sqrt{N} + ci/N$, by setting $\eta  = 1/\sqrt{N}$ in the above single-instance lower bound.

   The largest conditioned advantage the multi-instance adversary has is in the $k$-th round when the above value becomes $(1/2+1/\sqrt{N} +ck/N)$ for some constant $c$.

\end{proof}
\begin{corollary}
    By combining the above lemma with \Cref{lem:multi_instance_condition}, any polynomial-query adversary's advantage in our multi-instance game is $
    (1/2+\sqrt{1/N} +ck/N)^k$.
\end{corollary}

\paragraph{Putting Things Together}
Combining with \Cref{claim:remove_witness_with_multi_instance} with proper parameters, we acquire that the $\QCMA$ verifier with polynomial size classical witness $W$ succeeds in the interactive game with advantage $1/2+\delta$ for some negligible $\delta$.

Now we analyze the parameters for \Cref{claim:remove_witness_with_multi_instance}. 
We want:
\[
    \frac{1}{2} + \delta \leq 2^{W/k} (\frac{1}{2} + \frac{1}{\sqrt{N}} + \frac{ck}{N}) \leq \frac{1}{2} + \negl(\log N)~.
\]
We will need $W/k$ to be very small so that $2^{W/k}$ is close to 1. 

Since $W/k$ is very small, we have $2^{W/k} \leq 1+ W/k$. Therefore, we have:
\begin{align*}
    2^{W/k} \Big (\frac{1}{2} + \frac{1}{\sqrt{N}} + \frac{ck}{N} \Big ) \leq  \frac{1}{2} + \frac{ck}{N} + \frac{W}{2k} + \text{ (exponentially small terms) }
\end{align*}
Since $W$ is polynomial in  $n = \log N$, we can let $k$ be a superpolynomial in $n$
and the above value will be $\frac{1}{2} + \negl(\log N)$. 

\subsection{Distributional oracle version of interactive game}
We can now show that the interactive game involving a random challenge and a fixed oracle gives a separation for the distributional oracle game defined in \cite{natarajan2023distribution}.

There's a minor difference in the access of distributional oracle defined in \cite{natarajan2023distribution} and our definition of the interactive oracle game: the prover gets to access a set over Yes-instance oracle sets $\cD_{yes}$ arbitrarily (or to the No-instance oracle sets $\cD_{no}$ in the No case) in the following sense: whenever the prover makes a query on the ensemble of oracles $\cD_{yes}$ (or $\cD_{no}$ in the No-case), it is given access to an oracle $F(r, \cdot)$ for some uniform $r \gets Z$ (or $r \gets [R]/Z$ respectively in the No-case). 
On the other hand, in our interactive oracle distinguishing game, the prover gets arbitrary access to the set of both Yes and No instance oracles.
For $\QMA$ upper bound, the $\QMA$ witness and verifier agorithm in this case is the same as the one described in the interactive oracle distinguishing game \Cref{claim:qma_correctness_interactive}.

%
\paragraph{$\QCMA$ Lower Bound} We show a reduction from a $\QCMA$ verifier $V'$ that satisfies completeness and soundness in the distributional oracle game, to a $\QCMA$ verifier $V^*$ that satisfies completeness and soundness in the interactive game \Cref{sec:interactive_oracle_distinguish}. %
%
%
We can consider the following reduction: if there exists a $\QCMA$ prover-verifier pair $P', V'$ in the distributional oracle game such that $V'$ can perform a correct and sound and verification given $P'$'s classical witness, then there exists such a $\QCMA$ prover-verifier pair $P, V$ in the interactive oracle distinguishing game with the same correctness and soundness.

The prover $P$ in the interactive oracle distinguishing game receives oracle access to $F(\cdot, \cdot)$. 
it can simulate the distributional oracle  in the following way:
\begin{enumerate}
    \item Since $P$ does not know whether the oracle $F(r^*, \cdot)$ prepared later by the interactive oracle challenger, and since $P$ has to simulate distribution over oracles for $P'$ either over the Yes-instances or the No-instances, it guesses a uniformly random case Yes/No and tries to simulate this case for $P'$.
    
     If it guesses Yes: $P$ simulates $F(r, \cdot)$ for random $r \in Z$  by doing rejection sampling till it finds some $r$ that belongs to $Z$. $P$ can do so because it has unbounded computation power.
    Similarly, it can do so in the No-case, for random $F(r, \cdot), r \notin Z$. The distribution produced is the same as uniformly sampling over the Yes or No oracles
    respectively.
    
 After the distributional prover $P'$ produces a classical witness $w$, $P$ feeds it to its $\QCMA$ verifier $V$. $P$ also tells $V$ its guess on the challenge instance being Yes/No.

    \item     In the challenge phase, $V$ receives challenge oracle $F(r^*, \cdot)$ and runs $V'$ witness $w$. To decide whether it is Yes/No: if the guess $P$ tells $V$ matches $V'(w)$'s output, then outputs $V'(w)$'s output; if the guess $P$ tells $V$ does not match $V'(w)$'s output, then $V$ outputs a random guess Yes/No on its own. We assume that the distributional $V'(w)$ outputs the correct answer with probability $(1-\negl(n))$ when $P'$ is given the correct distribution of oracles (which can always be boosted since $w'$ is classical).
    Note that in both Yes and No cases,  the probability that $V(w)$ outputs the \emph{correct answer} is lower bounded by $\Pr[V'(w) \text{ is correct } \vert P \text{ guesses correctly }] \cdot \Pr[P \text{ guesses correctly }] + \frac{1}{2} \cdot \Pr[P \text{ guesses incorrectly }] = 1/2-\negl(n)+1/4 = 3/4 -\negl(n)$.  In the yes case, the correctness is $3/4 -\negl(n)$ and in the No case, the soundness error is $1/4+\negl(n)$. This would contradict our proven separation for the interactive game.
\end{enumerate}

\bibliographystyle{alpha}
\bibliography{qma}

\end{document}